\documentclass{article}
\usepackage{etex} 
\reserveinserts{28}

\usepackage[T1]{fontenc}
\usepackage[utf8]{inputenc}

\usepackage{microtype}

\usepackage{bbm}
\usepackage{mathrsfs}
\usepackage{latexsym}
\usepackage{paralist}

\usepackage{suffix}

\usepackage{xpunctuate}
\usepackage[all,british]{foreign} 
\usepackage{xfrac} 
\usepackage[style=english,english=british]{csquotes}

\usepackage{amsmath}
\usepackage{amsfonts,amssymb,amstext,amsthm} 
\usepackage{thmtools, thm-restate}
\usepackage{mathtools}

\usepackage[full,small]{complexity}

\usepackage[notref,color,final]{showkeys} 

\usepackage{xargs}
\usepackage{afterpage} 

\PassOptionsToPackage{dvipsnames,usenames,table}{xcolor}
\usepackage{tikz}
\usetikzlibrary{arrows,automata}
\usetikzlibrary{calc}
\usetikzlibrary{backgrounds}
\pgfdeclarelayer{foreground}
\pgfsetlayers{background,main,foreground}

\usepackage{booktabs} 
\usepackage{longtable}
\usepackage{float}
\usepackage{framed}
\usepackage{dcolumn}
\usepackage{adjustbox}

\usepackage{enumitem}
\usepackage{xspace}
\usepackage{xifthen}
\usepackage{fixltx2e}
\usepackage{url}
\usepackage{scrtime}

\usepackage[longend,vlined]{algorithm2e}

\usepackage{xkvltxp}
\usepackage[footnote,english,silent,draft]{fixme}
\newcommand{\todo}[1]{\fxfatal{\color{red}#1}}

\makeatletter
\newcommand{\raisemath}[1]{\mathpalette{\raisem@th{#1}}}
\newcommand{\raisem@th}[3]{\raisebox{#1}{$#2#3$}}
\makeatother

\def\N{\mathbb{N}} 

\newcommand\Problem[1]{\textsc{#1}\xspace}

\def\twobytwo(#1,#2,#3,#4){\ensuremath{\bigl( \begin{smallmatrix}
    #1  &  #2 \\%
    #3  &  #4%
  \end{smallmatrix} \bigr)}%
}

\newcommand\restr[2]{{
  \left.\kern-\nulldelimiterspace 
  #1 
  \vphantom{\big|} 
  \right|_{#2} 
  }}

\def\vminor{\preccurlyeq_{{\mathsf{v}}}\!} 
\WithSuffix\def\vminor^#1{%
  \preccurlyeq^{#1}_{{\mathsf{v}}}
} 
\def\pminor{\preccurlyeq_{{\mathsf{p}}}\!} 
\WithSuffix\def\pminor^#1{%
  \preccurlyeq^{#1}_{{\mathsf{p}}}
} 
\def\bminor_#1{\preccurlyeq_{{#1}}\!} 
\def\sminor^#1{
    \preccurlyeq_{\mathrlap{\mathsf{m}}}^{#1}
} 
\def\ssminor^#1{%
    \mathbin{\dot\preccurlyeq_{{\mathrlap{\mathsf{m}}}}^{#1}}\,
} 
\def\stminor^#1{
    \preccurlyeq_{\mathrlap{\mathsf{t}}}^{#1}
} 
\def\sstminor^#1{%
    \mathbin{\dot\preccurlyeq_{{\mathrlap{\mathsf{t}}}}^{#1}}\,
} 

\def\grad_#1{\nabla\!_{#1}}
\def\sgrad_#1{{\dot\nabla}\!_{#1}}
\def\topgrad_#1{\widetilde \nabla\!_{#1}}
\def\stopgrad_#1{\dot{\widetilde\nabla}\!_{#1}}
\def\topomega_#1{\widetilde \omega_{#1}}

\def\bnabla{\raisebox{-0.2pt}{\scalebox{1.25}{$\blacktriangledown$}}}
\def\vgrad_#1{\bnabla\!_{#1}}
\def\pgrad_#1{{\widetilde \bnabla}\!_{#1}}

\newcommand{\dist}{\ensuremath{\text{dist}}}

\def\colnum_#1{ \operatorname{col}_{#1} }

\def\wcolnum_#1{ \operatorname{wcol}_{#1} }
\def\adm_#1{ \operatorname{adm}_{#1} }

\renewcommand{\leq}{\leqslant}

\renewcommand{\geq}{\geqslant}

\renewcommand{\epsilon}{\varepsilon}

\newlength{\convarrowwidth}
\usepackage{stmaryrd}

\newcommand{\widthm}[1]{ \mathbf{#1} } 

\DeclareMathOperator{\width}{ \widthm{width} } 

\def\ds{\operatorname{\mathbf{dom}}}
\def\is{\operatorname{\mathbf{ind}}}
\def\scat{\operatorname{\mathbf{sct}}}

\usepackage{pifont}
\newcommand{\yaay}{\kern4pt \ding{51} \kern-8pt \ding{51}}%
\newenvironment{smallgraph}{
    \begin{tikzpicture}[-,auto,node distance=1.5cm,
                    semithick,every node/.style={transform shape}]%
    \tikzstyle{every node}=[scale=.7,fill=white,draw=darkgray,text=darkgray,align=center,initial text=, shape=circle]%
}{%
    \end{tikzpicture}
}

\definecolor{mygray}{HTML}{dddddd}

\declaretheoremstyle[%
spaceabove=6pt, spacebelow=6pt,
headfont=\normalfont\bfseries,
notefont=\mdseries, notebraces={(}{)},
bodyfont=\normalfont,
postheadspace=1em,
qed=\qedsymbol,
]{mystyle}
\theoremstyle{mystyle}

\declaretheorem{lemma}

\declaretheorem[name={Theorem}]{theorem}
\declaretheorem[name={Corollary}]{corollary}

\declaretheorem[numbered=no]{claim}
\declaretheorem[name={Definition}]{definition}

\declaretheoremstyle[%
spaceabove=6pt, spacebelow=6pt,
headfont=\normalfont\bfseries,
notefont=\mdseries, notebraces={\!\!\textbf{:}\,\,}{},
bodyfont=\normalfont,
postheadspace=1em,
headpunct=\mdseries.,
qed=,
]{case}
\theoremstyle{case}

\declaretheorem[name={Case},style=case]{case}

\renewcommand*\etal{\xperiodafter{\emph{et~al}}}

\newcommand*\varrule[1][0.4pt]{\leavevmode\leaders\hrule height#1\hfill\kern0pt}
\def\halfrule{\leavevmode\leaders\hrule height 0.7ex depth \dimexpr0.4pt-0.7ex\hfill\kern0pt} 

\setlist[1]{labelindent=\parindent,leftmargin=*}
\setlist{itemsep=0pt}
\setitemize[1]{label={\small\textbullet}}

\newenvironment{tightcenter}
 {\parskip=0pt\par\nopagebreak\centering}
 {\par\noindent\ignorespacesafterend}

\usepackage{ctable}
\newlength{\RoundedBoxWidth}
\newsavebox{\GrayRoundedBox}
\newenvironment{GrayBox}[1]%
   {\setlength{\RoundedBoxWidth}{\textwidth-4.5ex}
    \def\boxheading{#1}
    \begin{lrbox}{\GrayRoundedBox}
       \begin{minipage}{\RoundedBoxWidth}%
   }{%
       \end{minipage}
    \end{lrbox}%
    \begin{tightcenter}%
    \begin{tikzpicture}%
       \node(Text)[draw=black!20,fill=white,rounded corners,%
             inner sep=2ex,text width=\RoundedBoxWidth]%
             {\usebox{\GrayRoundedBox}};
        \coordinate(x) at (current bounding box.north west);
        \node [draw=white,rectangle,inner sep=3pt,anchor=north west,fill=white]
        at ($(x)+(6pt,.75em)$) {\boxheading};
    \end{tikzpicture}
    \end{tightcenter}\vspace{0pt}%
    \ignorespacesafterend
}

   {\setlength{\RoundedBoxWidth}{.8\textwidth-4.5ex}
    \def\boxheading{#1}
    \begin{lrbox}{\GrayRoundedBox}
       \begin{minipage}{\RoundedBoxWidth}%
   }{%
       \end{minipage}
    \end{lrbox}%
    \begin{tightcenter}%
    \begin{tikzpicture}%
       \node(Text)[draw=black!20,fill=white,rounded corners,%
             inner sep=2ex,text width=\RoundedBoxWidth]%
             {\usebox{\GrayRoundedBox}};
        \coordinate(x) at (current bounding box.north west);
        \node [draw=white,rectangle,inner sep=3pt,anchor=north west,fill=white]
        at ($(x)+(6pt,.75em)$) {\boxheading};
    \end{tikzpicture}
    \end{tightcenter}\vspace{0pt}%
    \ignorespacesafterend
}

\newenvironment{problem}[2][]{\noindent\ignorespaces%
                                \FrameSep=6pt%
                                \parindent=0pt%
                \vspace*{-.5em}
                \ifthenelse{\isempty{#1}}{%
                  \begin{GrayBox}{\textsc{#2}}%
                }{%
                  \begin{GrayBox}{\textsc{#2} parametrised by~{#1}}%
                }
                \newcommand\Prob{Problem:}%
                \newcommand\Input{Input:}%
                \begin{tabular*}{\textwidth}{@{\hspace{.1em}} >{\itshape} p{1.6cm} p{0.8\textwidth} @{}}%
            }{
                \end{tabular*}%
                \end{GrayBox}%
                \vspace*{-.5em}
                \ignorespacesafterend
            }

\newcounter{constraintnumbers}%
\newcommand\constrnumber{\stepcounter{constraintnumbers}(\arabic{constraintnumbers})}%

\renewenvironmentx{leftbar}[2][1=0.5pt, 2=5pt]{%
  \def\FrameCommand{\vrule width #1 \hspace{#2}}\MakeFramed {\advance\hsize-\width \FrameRestore}}%
{\endMakeFramed}

\newlength{\wleft}  \newlength{\wright}

\makeatletter
\let\orgdescriptionlabel\descriptionlabel
\def\@savelabel{}
\renewcommand*{\descriptionlabel}[1]{%
  \let\orglabel\label
  \let\label\@gobble
  \phantomsection
  \def\@savelabel{#1}
  \edef\@currentlabel{{\def\hfil{}#1}}
  \edef\@currentlabelname{#1}%
  \let\label\orglabel
  \orgdescriptionlabel{#1}%
}
\makeatother
\makeatletter
\def\namedlabel#1#2{\begingroup
   \def\@currentlabel{#1}%
   \label{#2}\endgroup
}
\makeatother

\usepackage{hyphenat}
\renewcommand{\th}{%
    \ifmmode
        ^\mathrm{th}%
    \else%
        \textsuperscript{th}\xspace%
    \fi%
}
\newcommand{\st}{%
    \ifmmode
        ^\mathrm{st}%
    \else%
        \textsuperscript{st}\xspace%
    \fi%
}
\newcommand{\nd}{%
    \ifmmode
        ^\mathrm{nd}%
    \else%
        \textsuperscript{nd}\xspace%
    \fi%
}
\newcommand{\rd}{%
    \ifmmode
        ^\mathrm{rd}%
    \else%
        \textsuperscript{rd}\xspace%
    \fi%
}

\def\Nesetril{Ne\v{s}et\v{r}il\xspace}

\def\Dvorak{Dvo\v{r}\'{a}k\xspace}

\usepackage{bbm}

\newcommand{\const}{\mathsf{c}}
\newcommand{\cproj}[1][r]{\ensuremath{\const^\text{proj}_{#1}}}
\newcommand{\cprojclos}[1][r]{\ensuremath{\const^\text{projcl}_{#1}}}
\newcommand{\cpathclos}[1][r]{\ensuremath{\const^\text{pathcl}_{#1}}}
\newcommand{\ctotalclos}[1][r,c]{\ensuremath{\const^\text{total}_{#1}}}
\newcommand{\cuqw}[1][d]{\ensuremath{\const^\text{UQW}_{#1}}}

\newcommand{\clilyscale}[1][r,d]{\ensuremath{\const^\text{scale}_{#1}}}
\newcommand{\clilymargin}[1][r,d]{\ensuremath{\const^\text{margin}_{#1}}}
\newcommand{\clilybase}[1][r,d]{\ensuremath{\const^\text{base}_{#1}}}

\newcommand{\clily}[1][]{\ensuremath{\const^\text{lily}_{#1}}}

\newcommand{\cdvorak}[1][r]{\ensuremath{\const^\text{dvrk}_{#1}}}

\newcommand{\ccdom}[1][r]{\ensuremath{\const^\text{cdom}_{#1,c}}}

\renewcommand{\G}{\mathcal G}

\newcommand{\ARCDS}{\textsc{Annotated $(r,c)$-Dominating Set}\xspace}
\newcommand{\RDS}{\textsc{$r$-Dominating Set}\xspace}
\newcommand{\rds}{$r$-dominating set\xspace}
\newcommand{\RCDS}{\textsc{$(r,c)$-Dominating Set}\xspace}
\newcommand{\rcds}{$(r,c)$-dominating set\xspace}

\newcommand{\ARCSC}{\textsc{Annotated $(r,c)$-Scattered Set}\xspace}
\newcommand{\RCSC}{\textsc{$(r,c)$-Scattered Set}\xspace}
\newcommand{\rcsc}{$(r,c)$-scattered set\xspace}

\newcommand{\LMD}{\textsc{$(r,[\lambda,\mu])$-Domination}\xspace}

\newcommand{\RPC}{\textsc{$r$-Perfect Code}\xspace}
\newcommand{\RRD}{\textsc{$r$-Roman Domination}\xspace}
\newcommand{\ARRD}{\textsc{Annotated $r$-Roman Domination}\xspace}
\newcommand{\RTD}{\textsc{Total $r$-Domination}\xspace}
\newcommand{\rtd}{total $r$-dominating set\xspace}
\newcommand{\ARTD}{\textsc{Annotated Total $r$-Domination}\xspace}

\newcommand{\simuni}{\ensuremath{\sim_{\nu}}}

\newcommand{\shad}{S}
\newcommand{\proj}{P}
\newcommand{\vol}{S\kern-1pt P}

\title{A general kernelization technique for domination and independence problems in sparse classes}

\author{
  Carl Einarson
  \thanks{
    Royal Holloway, University of London, UK \texttt{einarson.carl@gmail.com}.
  }
  \and
  Felix Reidl
  \thanks{
    Birkbeck, University of London, UK \texttt{f.reidl@dcs.bbk.ac.uk}
  }
}

\begin{document}

\maketitle

\begin{abstract}\noindent
  We unify and extend previous kernelization techniques in sparse
  classes~\cite{DSKernel, ISKernel} by defining \emph{water
  lilies} and show how they can be used in bounded expansion classes
  to construct linear bikernels for \RCDS,
  \RCSC, \RTD, \RRD, and a problem we call \LMD (implying a bikernel for \RPC).
  At the cost of slightly changing the output graph class our bikernels can be
  turned into kernels.

  We further demonstrate how these constructions can be combined to create `multikernels',
  meaning graphs that represent kernels for multiple problems at once. Concretely,
  we show that \RDS, \RTD, and \RRD admit a multikernel; as well as \RDS
  and \textsc{$2r$-Independent Set} for multiple values of~$r$ at once. 
\end{abstract}


\noindent
\textsc{Dominating Set} is arguably one of the touchstone for kernelization
in sparse graph classes: after a linear
kernel in planar graphs~\cite{DSKernelPlanar} and a polynomial kernel in
graphs defined by an excluded topological
minor~\cite{DSKernelHtopPoly,DSKernelHtopPoly2} results for linear kernels
in bounded genus graphs~\cite{DSKernelGenus} apex-minor-free
graphs~\cite{DSKernelApexMinor}, $H$-minor-free graphs~\cite{DSKernelHMinor},
and finally $H$-topological-minor-free graphs~\cite{DSKernelHTopMinor}
followed in quick succession. The most general results to date are linear
kernels for bounded expansion classes~\cite{DSKernel} (generalizing all
aforementioned classes) and an almost-linear kernels for nowhere dense
classes~\cite{DSKernelND} (generalizing bounded expansion classes). These
latter two results even hold for the general problem of
\RDS, where a vertex dominates everything in its closed $r$-neighbourhood.
Together with a recent almost-linear kernel for the related \textsc{$r$-Independence}
problem~\cite{ISKernel}, these results led us to the guiding question:
Do the kernelization techniques developed for
\textsc{$r$-Domination}/\textsc{$r$-Independence} in sparse classes
carry over to related problems? \looseness-1

\paragraph{Bounded expansion classes.}
\Nesetril and Ossona de Mendez introduced bounded expansion classes as a generalization
of classes excluding a (topological) minor and various useful notions of
sparsity (\eg embeddability in a surface, bounded degree). In short, a
class~$\mathcal G$ has \emph{bounded expansion} (BE) if any minor obtained by contracting
disjoint subgraphs of radius at most~$r$ in any member $G \in \mathcal G$
is $\grad_r(\mathcal G)$-degenerate, where~$\grad_r(\mathcal G)$ is a class constant
independent of $G$. There are various equivalent
definitions for BE classes~\cite{BndExp, WColBndExp, NBComplexity, Sparsity},
all of which have in common that they define families of graph
invariants~$\{f_r\}_{r \in \N}$ where~$r$ is a parameter governing the `depth'
at which the invariant is measured. BE classes then are precisely those graph
classes for which~$f_r$ is finite for every member of the class. We will not
need to work with these invariants directly, instead building on higher-level
results discussed in Section~\ref{sec:prelims}. Consequently, we broadly refer
to these invariants as \emph{expansion characteristics}. For an
in-depth discussion see~\cite{Sparsity}.


\paragraph{A selection of problems.}
The commonality of the following problems is that they can be expressed
via \emph{universal neighbourhood constraints}, meaning that a solution~$X$
needs to intersect every `neighbourhood' (a slightly flexible term as we
will see in the following) in at least/at most a certain value.

We define an \emph{$r$-dominating set} of a graph~$G$ to be any set~$D$
that satisfies $|N^r[u] \cap D| \geq 1$ for all~$u \in V(G)$, where
$N^r[u]$ contains all vertices at distance~$\leq r$ from~$u$.
We arrive at a natural extension of the problem by replacing the right
hand side of this \emph{domination constraint} by an arbitrary constant.
We call a set that satisfies the constraint $|N^r[u] \cap D| \geq c$
an \emph{\rcds} and the corresponding decision problem 

\begin{problem}[k]{$(r, c)$-Domination}
  \Input & A graph~$G$ and an integer~$k$. \\
  \Prob  & Is there a set~$D \subseteq V(G)$ of size at most~$k$ such that
           $|N^r[v] \cap D| \geq c$ for all $v \in G$?
\end{problem}

\noindent
For~$r = 1$ this problem has received some attention in the literature
under the name ``\textsc{$k$-Domination}'' (see \eg \cite{cDomInd}), for
$c = 1$ we recover the above discussed problem \textsc{$r$-Domination}.

We obtain a slightly different notion of dominance by insisting that vertices
cannot dominate themselves, but only their neighbourhood. It is natural to
extends this notion of \emph{total domination} by extending the domination
radius to some constant~$r$:

\begin{problem}[k]{Total $r$-Domination}
  \Input & A graph~$G$ and an integer~$k$. \\
  \Prob  & Is there sets~$D$, $|D| \leq k$ such that for every vertex~$v \in G$
           $|(N^r[v]\setminus\{v\}) \cap D| \geq 1$?
\end{problem}

\noindent
Finally, we might think of variants in which domination can occur at
different cost. One such variant is \textsc{Roman Domination} where we can
either pay one unit to let a vertex dominate itself (but not its neighbours)
or two units to dominate a vertex and its neighbourhood. We propose the following
generalization by allowing domination at distance~$r$:

\begin{problem}[k]{$r$-Roman Domination}
  \Input & A graph~$G$, a set~$L \subseteq V(G)$ and an integer~$k$. \\
  \Prob  & Is there sets~$D_1, D_2 \subseteq V(G)$ with $|D_1| + 2|D_2| \leq k$
           such that $D_2$ $r$-dominates all of~$V(G) \setminus D_1$?
\end{problem}

\noindent
While Roman domination does not quite fit the mould of universal neighbourhoods
constrains (since we can let vertices `opt out' of the constraint
$|N^r[v] \cap D_2| \geq 1$) this deviated is easily encompassed by our kernelization
technique.

The problem of \emph{independence} turns out to be closely related to that
of domination. We define an \emph{$r$-scattered set} of a graph~$G$ to be any set~$I$
that satisfies $|N^r[u] \cap I| \leq 1$ for all~$u \in V(G)$. Note that an $r$-scattered set is
equivalent to a \emph{$2r$-independent set} (all vertices in~$I$ are
pairwise at distance~$> 2r$) and the domination/independence duality
that holds in BE-classes (see below) has usually been described with this
terminology. However, the natural extension to \emph{$(r,c)$-scattered sets} that
satisfy the \emph{scatter constraints}~$|N^r[u] \cap I| \leq c$ does not correspond to
independent sets. We therefore opt to speak in terms of scattered instead of
independent sets, in particular, we consider the following parameterized problem:

\begin{problem}[k]{$(r,c)$-Scattered Set}
  \Input & A graph~$G$ and an integer~$k$. \\
  \Prob  & Is there a set~$I \subseteq V(G)$, $|I| \geq k$ such that
           $|N^r[v] \cap I| \leq c$ for all $v \in V(G)$?
\end{problem}

\noindent
Finally, we consider the problem that arises when combining the domination-
and scatter-constraints into the form $\lambda \leq |N^r[u] \cap D| \leq \mu$,
which leads to the following, rather general, parameterized problem:

\begin{problem}[k]{$(r,[\lambda, \mu])$-Domination}
  \Input & A graph~$G$ and an integer~$k$. \\
  \Prob  & Is there a set~$D \subseteq V(G)$ of size at most~$k$ such that
           every vertex~$v \in G$ satisfies
           $\lambda \leq |N^r[v] \cap D| \leq \mu$?
\end{problem}

\noindent
Here \textsc{$(r,[c,\infty])$-Domination} is equivalent to \RCDS and
\textsc{$(r,[0,c])$-Domination} to \RCSC. The problem further covers well-established
problems like \textsc{Perfect Code} which we again generalize by introducing
a distance-parameter:

\begin{problem}[k]{$r$-Perfect Code}
  \Input & A graph~$G$ an integer~$k$. \\
  \Prob  & Is there a set~$I \subseteq V(G)$ of size at most~$k$ such that
          \newline $|N^r[v] \cap I| = 1$ for all $v \in V(G)$?
\end{problem}

\vspace*{-1em}

\paragraph{Kernelization in sparse classes.}
The definition of a kernel (see~\cite{DowneyFellows} for a problem restricted to a certain
input class demands that the output belongs to this class as
well, \eg a planar kernelization needs to output a planar graph.
This turns out to be too restrictive for very general notions of sparseness
and we are left with the choice of either outputting an annotated
instance belonging to a different problem, called a \emph{bikernel}, or to modify the graph to
`simulate' the annotation in the original problem, but these modifications take
the instance out of the original graph class. Here we settle for the following
compromise: a parametrised graph problem $\mathcal P \subseteq \mathcal G \times \N$
for a BE-class $\mathcal G$ admits a \emph{BE kernel}
if there is a kernelization that outputs an instance in $\mathcal
G' \times \N$ with
$\grad_r(\mathcal G') \leq g(\grad_r(\G))$ for some
function~$g$ and all~$r \in \N$. This is justified by the idea that all nice
algorithmic properties stemming from $\mathcal G$ being BE
carry over from $\mathcal G$ to $\mathcal G'$ with only changes
to some constants\footnote{
  Using this argument, we might as well allow a BE-kernel to change
  the depth as well, \ie
   $\grad_r(\mathcal G') \leq g(\grad_{g(r)}(\G))$. In this work
   we do not need this level of generality and stick to a simpler
   definition.
}---if other properties of the class are of primary interest
(embedding in a surface, excluded minors, \etc) then the BE-view is simply too
coarse.

\vspace*{-.5em}
\paragraph{Our results.}
Inspired by the kernelization for \RDS~\cite{DSKernel} and
\textsc{$r$-Independent Set}~\cite{ISKernel} in sparse classes, we unify and
extend these techniques by defining a structure we call \emph{water lilies}
and show how their existence can be used to find small \emph{cores}, that is,
subset of vertices that either are guaranteed to contain a solution (\emph{solution core})
or that already fully represent the neighbourhood-constraints governing the
problem (\emph{constraint core}). We define and prove the existence of
water lilies in BE-classes in Section~\ref{sec:lily}, building on our proof
of a constant-factor approximation for \RCDS in BE-classes from Section~\ref{sec:approx-rcds}.\looseness-1

In Section~\ref{sec:bikernels} we use
water lilies to prove linear bikernels for \RCDS,
\RCSC, \RTD, \RRD, and \LMD (implying a bikernel for \RPC) into appropriate
annotated variants of these problems. We then show in Section~\ref{sec:be-kernels}
how these bikernels can be turned into BE-kernels for
\RCDS, \RCSC, \RTD, \RRD, and \RPC. 
Finally, in Section~\ref{sec:multikernels},
we demonstrate how these constructions can be combined to create `multikernels',
meaning graphs that represent kernels for multiple problems at once. Concretely,
we show that \RDS, \RTD, and \RRD admit a multikernel; as well as \RDS
and \textsc{$2r$-Independent Set} (even for multiple values of~$r$ at once).

\section{Notation and previous results}\label{sec:prelims}

For a maximization problem $\mathcal P$ defined via universal neighbourhood
constraints and a graph~$G$ we call a set~$L \subseteq V(G)$ a \emph{constraint
core} if for every set~$D \subseteq V(G)$ it holds that $D$ is a solution
to~$\mathcal  P$ in~$G$ already if the constraints only hold for vertices
in~$L$. Analogous, for a minimization problem~$\mathcal P$ defined via universal
neighbourhood constraints, we call a set~$U \subseteq V(G)$ a \emph{solution
core} if a minimum solution to~$P$ already exists inside~$U$. In both cases,
note that~$V(G)$ is always a trivial core and that a superset of any core is a
core as well.

A set $D \subseteq V(G)$ is an \emph{\rcds} if for every vertex $v \in V(G)$
it holds that $|N^r[v] \cap D |\geq c$. Importantly, this constraint must also
hold for vertices contained in $D$, therefore such a set can only exist if
$|N^r[v]| \geq c$ for all $v \in G$. We write $\ds^c_r(G)$ to denote the size
of a minimum \rcds in $G$ and let $\ds^c_r(G) = \infty$ if no such set exists.

A set~$I \subseteq V(G)$ is \emph{$2r$-independent} if every pair of
vertices $u, v \in I$ has distance at least~$2r+1$. We write $\is_{2r}(G)$
to denote the size of a maximum $2r$-independent set in $G$. Related,
a set $I \subseteq V(G)$ is an \emph{\rcsc} if for all vertices~$v \in G$
it holds that $|N^r[v] \cap I| \leq c$. An $(r,1)$-scattered set is equivalent
to a $2r$-independent set, but this relationship breaks down for $c > 1$.
 We defined~$\scat^c_r(G)$ as the size of a maximum \rcsc in~$G$.
In all cases, for~$c = 1$ we will omit the superscript.

In many of the following constructions we will use the phrase ``connect
$u$ to~$v$ by a path of length~$r$''. This operation is to be understood as
adding~${(r-1)}$ new vertices~$a_1,\ldots,a_{r-1}$ the graph and then adding
edges to create the path~$u a_1 \ldots a_{r-1} v$.


\subsection{Domination/independence duality}

We adapted the following results to use the notation introduced above for
the sake of a unified presentation. In particular, we will be using
$\scat_r$ instead of~$\is_{2r}$. The function $\wcolnum_r$ is one of the
expansion characteristics mentioned above (see \eg \cite{WColBndExp} for a
definition), here it is enough to know that for every member $G$ of a BE-class,
$\wcolnum_r(G)$ is bounded by a constant for every $r \in \N$.

\begin{theorem}[\Dvorak~\cite{DvorakDomset}]\label{thm:dvorak-duality}
  For every graph~$G$ and integer~$r \in \N$ it holds that
  \[
    \scat_r(G) \leq
    \ds_r(G) \leq \wcolnum_{2r}^2(G) \scat_r(G)
  \]
\end{theorem}

\noindent
\Dvorak recently showed an improved bound~\cite{DvorakDomset2}, we will use the
above simpler expression. In the same paper he further proved the following
relationship between $r$-scattered sets and $(r,c)$-scattered sets
(translated into our terminology):

\begin{theorem}[\Dvorak~\cite{DvorakDomset2}]\label{thm:dvorak-is}
  For every graph~$G$ and integers~$c,r \in \N$ it holds that
  \[
    \frac{1}{2c \wcolnum_{2r}(G)} \scat_r^c(G) \leq \scat_r(G) \leq \scat_r^c(G)
  \]
\end{theorem}

\begin{theorem}[\Dvorak's algorithm~\cite{DvorakDomset}]\label{thm:dvorak-ds}
  For every BE class $\G$ and $r \in \N$ there exists a constant $\cdvorak$
  and a polynomial-time algorithm that computes an $r$-dominating set $D$ of $G$
  and an $r$-scattered set $A \subseteq D$ with $|D| \leq \cdvorak |A|$.
\end{theorem}

\noindent
In particular, the $r$-scattered set $A$ witnesses that $D$ is indeed a
$\cdvorak$-approximation of a minimum $r$-dominating of $G$. This algorithm can
further be modified to compute a dominating set for a
specific set $X \subseteq V(G)$ only; in that case it outputs the sets $A$ and $D$,
$A \subseteq D \cap X$, where $D$ dominates all of $X$ in $G$ and $A$ is
$r$-scattered in $G$. We will call this algorithm the \emph{warm-start}
variant since we only need to mark the vertices $V(G)\setminus X$ as
already dominated and then run the original algorithm (an alternative is a small
gadget construction~\cite{DSKernel}).

\subsection{Projections, shadows, and distance preservation}

Given a vertex set $X \subseteq V(G)$ we call a path
\emph{$X$-avoiding} if its internal vertices are not contained
in $X$. A \emph{shortest $X$-avoiding path} between vertices $x,y$
is shortest among all $X$-avoiding paths between $x$ and $y$.

\begin{definition}[$r$-projection]
  For a vertex set $X \subseteq V(G)$ and a vertex $u \not \in X$
  we define the \emph{$r$-projection} of $u$ onto $X$ as the set
  \[
    \proj^r_X(u) := \{ v \in X \mid \text{there exists an $X$-avoiding $u$-$v$-path of length} \leq r \}
  \]%
\end{definition}

\noindent
Note in particular that $\proj^1_X(u) = N(u) \cap X$, but for $r > 1$ the sets
$\proj^r_X(u)$ and $N^r(u) \cap X$ might differ.

\begin{definition}[$r$-shadow]
  For a vertex set $X \subseteq V(G)$ and a vertex $u \not \in X$
  we define the \emph{$r$-shadow} of $u$ onto $X$ as the set
  \[
    \shad^r_X(u) := \{ v \in V(G) \mid \text{every $u$-$v$-path of length} \leq r
    ~\text{has an internal vertex in~$X$} \}
  \]%
\end{definition}

\noindent
The shadow $\shad^r_X(u)$ contains precisely those vertices that are `cut off' by
the set $\proj^r_X(u)$. We will frequently need the union of shadow and projection
and therefore introduce the shorthand $\vol^r_X(u) := \shad^r_X(u) \cup \proj^r_X(u)$.

 Two vertices that have the same $r$-projection onto $X$
do not, however, necessarily have the same shadow since the precise distance
at which the projection lies might differ. To distinguish such cases, it
is useful to consider the \emph{projection profile} of a vertex to its projection:

\begin{definition}[$r$-projection profile]
  For a vertex set $X \subseteq V(G)$ and a vertex $u \not \in X$
  we define the \emph{$r$-projection profile} of $u$ wrt $X$ as a function
  $\pi^r_{G,X}[u] \colon X \to [r] \cup \infty$ where $\pi^r_{G,X}[u](v)$
  for $v \in X$ is the length of a shortest $X$-avoiding path from $u$ to $v$
  if such a path of length at most $r$ exists and $\infty$ otherwise.
\end{definition}

\noindent
We say that a function $\nu \colon X \to [r] \cup \infty$ is \emph{realized on
$X$} (as a projection profile) if there exists a vertex $u \not \in X$ for which
$\nu = \pi^r_{G,X}[u]$ and we denote the set of all realized profiles by $\Pi_G^r(X)$.
We will usually drop the subscript~$G$ if the graph is clear from the context.
It will be convenient to define an equivalence relation that
groups vertices outside of $X$ by their projection profile. Define
\[
  u \sim^r_{X} v \iff \pi^r_X[u] = \pi^r_X[v]
\]
for pairs $u,v \in V(G)\setminus X$.

It turns out that in BE classes, the number of possible projection profiles
realised on a set $X$  is bounded linearly in the size of $X$.

\begin{lemma}[Adapted from~\cite{DSKernel,DSKernelND}]\label{lemma:projbound}
  For every BE class $\G$ and $r \in \N$ there exists a constant $\cproj$
  such that for every $G \in \G$ and $X \subseteq V(G)$, the number of
  $r$-projection profiles realised on $X$ is at most $\cproj |X|$.
\end{lemma}

\noindent
In our notation this can alternatively be written as $|\Pi^r(X)| =
|(V(G)\setminus X) / {\sim^r_X}|  \leq \cproj |X|$. We will crucially
rely on the following two results for BE classes:

\begin{lemma}[Projection closure~\cite{DSKernel}]\label{lemma:projclos}
  For every BE class $\G$ and $r \in \N$ there exists a constant $\cprojclos$
  and a polynomial-time algorithm that, given $G \in \G$ and $X \subseteq V(G)$,
  computes a superset $X' \supseteq X$, $|X'| \leq \cprojclos |X|$, such that
  $|\proj^r_{X'}(u)| \leq \cprojclos$ for all $u \in V(G)\setminus X'$.
\end{lemma}

\begin{lemma}[Shortest path closure~\cite{DSKernel}]\label{lemma:pathclos}
  For every BE class $\G$ and $r \in \N$ there exists a constant $\cpathclos$
  and a polynomial-time algorithm that, given $G \in \G$ and $X \subseteq V(G)$,
  computes a superset $X' \supseteq X$, $|X'| \leq \cpathclos |X|$, such that
  for all $u, v \in X$ with $\dist(u,v) \leq r$ it holds that
  $\dist_{G[X']}(u,v) = \dist(u,v)$.
\end{lemma}

\noindent
It will be useful to combine the above two lemmas in the following way:

\begin{definition}[Projection kernel]\label{def:projkernel}
  Given a graph~$G$ and a set~$X \subseteq V(G)$, an
  \emph{$(r,c)$-projection kernel} of~$(G,X)$ is an induced subgraph~$\hat G$ of~$G$
  with~$X \subseteq V(\hat G$) and the following properties:
  \begin{enumerate}
    \item $N_{\hat G}^d(v) \cap X = N_{G}^d(v) \cap X$ for all~$v \in X$ and~$d \leq r$; and
    \item if the signature~$\nu \colon X \to [r] \cup \infty$ is realized on~$X$ by~$p$ distinct
    vertices in~$G$, then~$\nu$ is realized by at least~$\min\{c,p\}$ distinct
    vertices in~$\hat G$.
  \end{enumerate}
\end{definition}

\begin{lemma}\label{lemma:projkernel}
  For every BE class $\G$ and $c, r \in \N$ there exists a constant $\ctotalclos$
  and a polynomial-time algorithm that, given $G \in \G$ and $X \subseteq V(G)$,
  computes an $(r,c)$-projection kernel~$\hat G$ of~$(G,X)$ with
  $|\hat G| \leq \ctotalclos |X|$.
\end{lemma}
\begin{proof}
  We first apply Lemma~\ref{lemma:projclos} to~$X$ and obtain a set~$X_1
  \supset X$, $|X_1| \leq \cprojclos[r] |X|$, such that the projections of
  outside vertices onto~$X_1$ have size at most~$\cprojclos[r]$.

  Next, we
  apply Lemma~\ref{lemma:pathclos} to~$X_1$ and receive a set~$X_2 \supset X_1$, $|X_2| \leq
  \cpathclos |X_1|$, such that the graph~$G[X_2]$ preserves short distances
  (less than or equal to~$r$) between vertices in~$X_1$. Finally, let~$U$
  contain up to~$c$ representatives for every equivalence class~$[u] \in V(G)
  / \sim^r_{X_1}$ (if the class is smaller than~$c$ we include all of it). By
  Lemma~\ref{lemma:projbound} we have that $|U| \leq c \cdot \cproj |X_1|$.

  Construct now~$X_3$ by taking the union~$X_2 \cup U$ as well
  as shortest paths from every member~$u \in X_2 \cup U$ to all of
  $\proj^r_{X_1}(u)$. By definition, each of these paths has length at most~$r$
  and therefore contains at most~$r-1$ internal vertices. Since,
  by construction of~$X_1$, $|\proj^r_{X_1}(u)| \leq \cproj$; it follows that we
  add at most~$\cproj (r-1)$ vertices per vertex in~$X_2 \cup U$.
  Taking the above bounds together, we have that
  \[
    |X_3| \leq (r-1)\, \cproj (\cpathclos + c \cdot \cproj) |X| =: \ctotalclos |X|.
  \]
  It remains to be shown that~$\hat G := G[X_3]$ has the desired properties.

  Property~1 follows directly from the fact that already~$G[X_2] \subseteq \hat G$
  preserves short distances among vertices inside~$X_1 \supseteq X$. In particular,
  each vertex in $X_1\setminus X$ has the same $r$-projection profile onto~$X$
  in~$G$ and~$\hat G$.

  To see that Property~2 holds, consider
  any profile~$\nu$ realized on~$X$ by vertices~$S \subseteq V(G) \setminus X$
  in~$G$. First consider the case $S \setminus X_1 \neq \emptyset$.
  Then by construction, the set~$U$ contains $\min\{c,|S\setminus X_1|\}$ vertices
  from~$S\setminus X_1$ that realize $\nu$ in~$G$ and whose projection onto~$X_1$ is the
  same in~$G$ and~$\hat G$. Since~$X_1 \supseteq X$, we conclude that their
  projection on~$X$ in~$\hat G$ must be~$\nu$. By the above, the vertices
  in $S \cap X_1$ must have the profile~$\nu$ as well.
  Now assume $S \subseteq X_1$, therefore no vertex outside of~$X_1$ has the
  profile~$\nu$ in~$G$. As argued above, $S$ has the profile $\nu$ in~$\hat G$
  as well, therefore~$\hat G$ contains~$|S| \geq \min\{c, |S|\}$ vertices with profile~$\nu$,
  as claimed.
\end{proof}

\noindent
Note that the above construction implies that~$\Pi^r_{\hat G}(X) \supseteq \Pi^r_G(X)$,
however, it is not necessarily true that~$\Pi^r_{\hat G}(X) = \Pi^r_G(X)$.

The following is a slight restatement of Theorem~4 in \cite{BEGenColouring}.
We emphasise that the proof by Kreutzer \etal is actually constructive
and can be implemented to run in polynomial time. \looseness-1

\begin{lemma}[UQW in BE classes~\cite{BEGenColouring}]\label{lemma:UQW}
  For every BE class $\G$ and \emph{distance} $d \in \N$ there exists a constant $\cuqw$
  and a polynomial-time algorithm that, given $G \in \G$, a \emph{size} $t \in \N$ and $X
  \subseteq V(G)$ with $|X| \geq \cuqw \cdot 2^t$,
  computes a set $S$ of size at most $(\cuqw)^2$ and $X' \subseteq X \setminus S$
  of size at least $t$ such that $X'$ is $d$-scattered in $G - S$.
\end{lemma}


\section{Approximating \RCDS}\label{sec:approx-rcds}

%

\begin{theorem}\label{thm:cdomapprox}
  Let $\mathcal G$ be a BE class and fix $r, c \in \N$. There exists a
  constant $\ccdom$ and an algorithm that, for
  every $G \in \mathcal G$, computes in polynomial time an \rcds of size at most
  $
      \ccdom \ds^c_r(G)
  $ or concludes correctly that $G$ cannot be $(r,c$)-dominated.
\end{theorem}
\begin{proof}
  We compute a sequence of dominating sets~$D_1, D_2, \ldots, D_c$ with
  the invariants that a) $D_i$ $(r,i)$-dominates $G$ and b)
  $|D_{i+1}| \leq  5\cdvorak\cproj |D_i| + \cdvorak\ds^{i+1}_r(G)$.

  To start the process, let $D_1$ be an $\cdvorak$-approximate
  $r$-dominating set for $G$, this set clearly satisfies invariant a).
  We proceed in two steps to construct $D_{i+1}$ from $D_i$.
  Build the set $U_i$ as follows: for every projection $\mu \in \Pi^r(D_i)$
  realized by an equivalence class~$[v] \in (V(G)\setminus D_i)  /
  {\sim^r_{D_i}}$ we pick one (arbitrary) vertex from $\shad^r_{D_i}(v)
  \setminus D_i$ and add it to $U_i$, if such a vertex exists.
  Then for every vertex $u \in D_i$ that is not $(i+1)$-dominated by
  $D_i \cup U_i$, we add an arbitrary vertex from $N^r[u] \setminus D_i$
  to $U_i$ (note that if no such vertex exists we conclude that $G$ cannot be
  $(r,c)$-dominated).

  By construction, the size of $U_i$ is bounded by
  $|U_i| \leq |\Pi^r(D_i)| + |D_i| \leq (\cproj + 1) |D_i|$. Further
  note that every vertex in $D_i \cup U_i$ is $(r,i+1)$-dominated
  by $D_i \cup U_i$: due to
  invariant a), the set $D_i$ $(r,i)$-dominates $D_i \cup U_i$ and $U_i$
  now additionally dominates itself (at least) once and, by construction,
  those vertices in~$D_i$ that are not yet $(r,i+1)$-dominated by~$D_i$.

  Define the set $R_i$ to contain all vertices that are \emph{not} $(r,i+1)$-dominated
  by $D_i \cup U_i$, note that in particular~$N^r[R_i] \cap U_i = \emptyset$.
  Let $G' = G-(D_i\cup U_i)$.  Apply \Dvorak's warm-start
  algorithm to find a distance-$r$ dominator~$D_i'$ for $R_i$ in $G'$ and a
  $r$-scattered set $A_i' \subseteq D_i' \cap R_i$ with $|A_i'| \leq |D_i'| \leq
  \cdvorak |A_i'|$.

  \begin{claim}
      $|A_i'| \leq  (\cproj + 1) |D_i| + \ds^{i+1}_r(G)$.
  \end{claim}
  \begin{proof}
    Let $X$ be an $(r,i+1)$-dominating set of $G$ of minimum size and assume that
    $|A_i'| > (\cproj + 1) |D_i| + \ds^{i+1}_r(G) \geq |U_i \cup X|$. Then there
    exists $a \in A_i'$ such  that $N^r_{G'}[a] \cap (U_i \cup X) = \emptyset$. Since
    $X$ $(r,i+1)$-dominates $a$ but $D_i \cup U_i$ does not (because $a \in R_i$) there
    must be at least one vertex $b \in X \cap (N^r_G[a]\setminus N^r_{G'}[a])$ that
    is not contained in $D_i \cup U_i$. This means that~$b \in \shad^r_{D_i \cup U_i}(a)$
    and since $N^r_G[a] \cap
    U_i = \emptyset$, we have that $\shad^r_{D_i \cup U_i}(a) = \shad^r_{D_i}(a)$ and
    therefore even~$b \in \shad^r_{D_i}(a)$. But then,
    since $b \not \in D_i \cup U_i$, we could have added $b$ to $U_i$ during the
    first construction phase in order to dominate the class $[a]$. The existence
    of $a$ leads us to a contradiction and we conclude that
    $|A_i'| \leq  (\cproj + 1) |D_i| + \ds^{i+1}_r(G)$.
  \end{proof}

  \noindent
  Finally, construct the set $D_{i+1} = D_i' \cup D_i \cup U_i$. Since $D_i'$ $r$-dominates
  $R_i$ which, by construction, were the only vertices not yet $(r,i+1)$-dominated
  by $D_i \cup U_i$, we conclude that $D_{i+1}$ is indeed an $(r,i+1)$-dominating
  set of $G$; thus invariant a) is preserved. To see that invariant b) holds,
  let us bound the size of $D_{i+1}$:
  \begin{align*}
    |D_{i+1}| &\leq |D_i'| + |D_i| + |U_i|
      \leq \cdvorak |A_i'| + |D_i| + (\cproj + 1) |D_i| \\
      &\leq \cdvorak(\cproj + 1) |D_i| + \cdvorak\ds^{i+1}_r(G) + (\cproj + 2) |D_i| \\
      &= \big( (\cdvorak + 1)(\cproj + 1) + 1 \big) |D_i| + \cdvorak\ds^{i+1}_r(G) \\
      &\leq \big( 4\cdvorak\cproj + 1 \big) |D_i| + \cdvorak\ds^{i+1}_r(G) \\
      &\leq 5\cdvorak\cproj |D_i| + \cdvorak\ds^{i+1}_r(G).
  \end{align*}
  We conclude that invariant b) holds, as claimed. Resolving the recurrence
  provided by this inequality, we finally obtain the bound
  \begin{align*}
    |D_c| &\leq \cdvorak \sum_{i=1}^{c} (5\cdvorak\cproj)^{c-i} \ds^i_r(G)
          \leq (5\cdvorak\cproj)^{c+1} \ds^c_r(G),
  \end{align*}
  and the claim follows with $\ccdom := (5\cdvorak\cproj)^{c+1}$.
\end{proof}


\section{Water lilies}\label{sec:lily}

\begin{figure}[b]
  \centering\includegraphics[width=.9\textwidth]{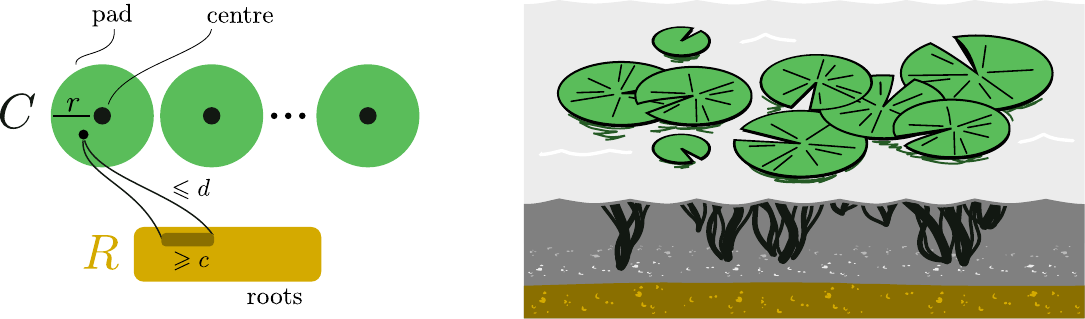}
  \caption{\label{fig:waterlily}%
    Schematic of a water lily $(R,C)$ with radius~$r$, depth~$d$ and
    adhesion~$c$. Removing the `tangled' roots~$R$ creates disjoint $r$-neighbourhoods 
    around~$C$ which we imagine like lily pads floating on a pond.
  }
\end{figure}

\begin{definition}[Water lily]
  A \emph{water lily} of \emph{radius~$r$}, \emph{depth~$d \leq r$}
  and \emph{adhesion~$c$}
  in a graph~$G$ is a tuple
  $(R, C)$ of disjoint vertex sets with the following properties:
  \begin{itemize}
    \item $C$ is $r$-scattered in $G-R$,
    \item $N^r_{G-R}[C]$ is $(d,c)$-dominated by $R$ in $G$.
  \end{itemize}
  We call $R$ the \emph{roots}, $C$ the \emph{centres}, and the sets $\{N^r_{G-R}[x]\}_{x \in C}$
  the \emph{pads} of the water lily.

  A water lily is \emph{uniform} if all centres have the same
  $d$-projection onto $R$, \eg $\pi^d_R[x]$ is the same function for
  all $x \in C$. The \emph{ratio} of a water lily is any guaranteed lower
  bound on $|C| / |R|$. \looseness-1
\end{definition}

\noindent
The following lemma lies at the heart of our unification of previous
techniques \cite{DSKernel, DSKernelND, ISKernel}. It streamlines the construction of BE-kernels considerably, as we will see in the
following section.

\begin{lemma}\label{lemma:lily}
  For every BE class $\G$ and $c,r,d \in \N$, $d \leq r$, there exist constants
  $\clilyscale[c,r,d]$, $\clilymargin[c,r,d]$, $\clilybase$ with the following property:
  for every $G \in \G$ which has an $(r,c)$-dominating set, $t \in \N$ and $A \subseteq V(G)$
  with $|A| \geq \clilyscale[c,r,d] \cdot (\clilybase)^t \cdot \ds^c_d(G)$ there
  exists a uniform water lily $(R, C)$, $C \subseteq A$, with depth~$d$, radius~$r$, adhesions~$c$
  and with $|R| \leq \clilymargin[c,r,d]$, $|C| \geq t$.
  Moreover, such a water lily can be computed in polynomial time.
\end{lemma}
\begin{proof}
  Given $G$, we use Theorem~\ref{thm:cdomapprox} to compute a
  $(d,c)$-dominating set $D'$ of size at most $\ccdom \cdot \ds_d(G)$ in
  polynomial time or conclude that no such set exits.  Afterwards, we compute
  the $(r+d)$-projection closure~$D$ of $D'$, by Lemma~\ref{lemma:projclos} we
  have that $|D| \leq
  \cprojclos[r+d] |D'|$ and thus $|D| \leq \cprojclos[r+d] \ccdom \ds_d(G)$.
  Let $A'' := A \setminus D$, we will later choose $\clilyscale[c,r,d]$ so that
  $A''$ is still large enough for the following arguments to go through.

  Define the equivalence relation $\sim_D$ over $A''$ via $a \sim_D a'
  \iff \pi^{r+d}_D[a] = \pi^{r+d}_D[a']$. By Lemma~\ref{lemma:projbound}, the number
  of classes in $A'' / \sim_D$ is bounded by $\cproj[r+d] |D|$; by an averaging argument
  we have at least one class $[a] \in A'' / \sim_D$ of size
  \[
    \big|[a] \big| \geq \frac{|A''|}{\cproj[r+d] |D|} \geq \frac{|A| - |D|}{\cproj[r+d] |D|}.
  \]
  Let $R'''$ be $\proj^{r+d}_D(a)$, \ie
  the $(r+d)$-projection of $[a]$'s members on $D$.
  By our earlier application of Lemma~\ref{lemma:projclos}
  we have that $|R'''| = |\proj^{r+d}_D(a)| \leq \cprojclos[r+d]$.

  Again, we will choose $\clilyscale[c,r,d]$ large enough to apply Lemma~\ref{lemma:UQW}
  with distance~$r$ and size $\cproj[d] |R'''| t$ to the set $[a]$ and receive a subset
  $A' \subseteq [a]$ of size at least $\cproj[d](\cuqw[r] + \cprojclos[r+d]) \cdot t$
  and a set $R'' \subseteq V(G)\setminus A'$,
  $|R''| \leq \cuqw[r]$, such that $A'$ is $r$-scattered in $G-R''$.
  Let $R' := R'' \cup R'''$, by the above bounds on $R''$ and $R'''$ it follows
  that $|R'| \leq \cuqw[r] + \cprojclos[r+d]$.
  By Lemma~\ref{lemma:projbound} and the fact that
  \[
    |A'| \geq \cproj[d](\cuqw[r] + \cprojclos[r+d]) \cdot t
         \geq |\Pi^d(R')| \cdot t
  \]
  there exists a set $C \subseteq A'$ of
  size at least $t$ such that all members of $C$ have the same $d$-projection
  onto $R'$.

  We construct the set $R$ from $R'$ as follows: for every projection
  profile $\mu \in \Pi^d(R')$ realized by a class $[u] \in N^r_{G-R'}[C] /
  {\sim^d_{R'}}$ we add $\max\{0, c - |\proj^d_{R'}(u)|\}$ vertices from the
  shadow $\shad^d_{R'}(u) \cap D'$. Since $D'$ $(d,c)$-dominates all of $G$,
  such vertices must exist. By construction, $|R| \leq c |R'|$ and $R$
  $(c,d)$-dominates all of $N^r_{G-R'}[C]$ and thus in particular
  $N^r_{G-R}[A']$. Note further that all vertices we added lie
  inside~$\shad^{r+d}_{R'}[C]$, therefore the projection profiles of~$C$ are
  not changed by this operation (all paths of length at most~$r+d$ from $C$ to
  vertices in $R/R'$ pass through $R'$). We conclude that the uniformity
  condition holds on $(R,C)$. This construction also provides us with the bound $|R|
  \leq c(\cuqw[r] + \cprojclos[r+d]) =: \clilymargin[c,r,d]$.

  Finally, let us determine a value for $\clilyscale[c, r,d]$ that suffices for the above
  construction to go through. In order to apply Lemma~\ref{lemma:UQW}, we
  need that $|[a]| \geq \cuqw[r] \cdot 2^{\cproj[d](\cuqw[r] + \cprojclos[r+d]) \cdot t}$, accordingly we need
  that
  \begin{align*}
    &\phantom{\impliedby}  \frac{|A| - |D|}{\cproj[r+d] |D|} \geq
    \cuqw[r] \cdot 2^{\cproj[d](\cuqw[r] + \cprojclos[r+d]) \cdot t } \\
    &\impliedby
      \frac{|A|}{2 \cproj[r+d] |D|} \geq \cuqw[r] \cdot 2^{\cproj[d](\cuqw[r] + \cprojclos[r+d]) \cdot t } \\
    &\impliedby
      \frac{|A|}{2 \cproj[r+d] \cprojclos[r+d] \ccdom \ds^c_d(G)} \geq \cuqw[r] \cdot 2^{\cproj[d](\cuqw[2r] + \cprojclos[r+d]) \cdot t }
  \end{align*}
  We conclude that choosing the constants $\clilyscale[c,r,d] = 2 \cproj[r+d] \cprojclos[r+d] \ccdom \cuqw[2r]$
  and $\clilybase = 2^{\cproj[d](\cuqw[r] + \cprojclos[r+d])}$ suffices to prove
  the claim.
\end{proof}


\noindent
We can impose even more structure on a water lily in the following sense: let
us define a \emph{pad signature} as a function~$\sigma\colon C \to \Sigma^*$
(for some alphabet~$\Sigma$) that can be computed by a polynomial-time
algorithm receiving the following inputs:
\begin{itemize}
  \item The depth~$d$, radius~$r$ and adhesion~$c$ of the water lily;
  \item the centre~$a$, its pad~$N^r_{G-R}[a]$, the roots~$R$;
  \item the subgraph~$G[R \cup N^r_{G-R}[a]]$ alongside potential vertex/edge
        labels from the host graph~$G$.
\end{itemize}
We say that $\sigma$ is \emph{bounded}
if the size of its image can be bounded by a constant.

Every pad signature~$\sigma$ gives rise to an equivalence relation
$\sim_\sigma \subseteq C \times C$ for a water lily~$(R,C)$ via
\[
  a \sim_\sigma a' \iff \sigma(a) = \sigma(a').
\]
Note that if~$\sigma$ is bounded, then~$\sim_\sigma$ has finite index.
A water lily is \emph{$\sigma$-uniform} if all its centres belong to the
same equivalence class under~$\sim_\sigma$; or alternatively if all centres
have the same image under~$\sigma$.
For a bounded signature~$\sigma$, we find a $\sim_\sigma$-uniform water lily of
ratio~$\tau$ by first finding a water lily~$(R', C')$ with ratio~$p \cdot \tau$,
where $p$ is an upper bound on the image of~$\sigma$, and then return $R'$
together with the largest class in $C' / \sim_\sigma$. Accordingly:

\begin{corollary}\label{cor:lily}
  For every BE class $\G$, $c,r,\tau \in \N$ and pad signature~$\sigma$ with
  finite index there exists a constant $\clily = \clily[c,2r,r,\tau,\sigma]$
  with the following property: for every $G \in \G$ which has an $(r,c)$-dominating
  set and $A \subseteq V(G)$
  with $|A| \geq \clily \cdot \ds^c_d(G)$ there exists a $\sigma$-uniform water
  lily $(R,C)$, $C \subseteq A$, $|R| \leq \clily$, of  depth $r$, radius~$2r$,
  adhesion~$c$ and ratio $\tau$.

  Moreover, such a water lily can be computed in polynomial time.
\end{corollary}

\noindent
Let us define a particular bounded pad signature that will be
useful in the remainder: define $\nu$ as
\[
  \nu(a) := (\{ \pi^d_R[x] \mid x \in N^i_{G-R}(a) \} \mid 0 \leq i \leq r \big),
\]
where the right-hand side is to be understood as encoded in a string
by some suitable scheme.
Two centres are equivalent under~$\simuni$ if they have the same
projection-types at the same distance (though potentially at different
multiplicities) inside their respective pads.
Since $|R|$ has
constant size according to Lemma~\ref{lemma:lily} and there are at most
$\cproj[d] |R|$ possible projection profiles according to
Lemma~\ref{lemma:projbound}, the image of~$\nu$ has size at most
 $r^{\cproj[d] |R|} \leq r^{\cproj[d] \clily}$ and therefore~$\nu$ is
a bounded pad signature.

We will sometimes combine $\nu$ with a finite number of vertex
labels that arise during the construction of bikernels. If vertices are
labelled by~$f\colon V(G)\to \Sigma$ for some finite alphabet~$\Sigma$, 
then we understand $\nu$ to be the above equivalence relation further 
refined by the equivalence relation $u \sim_f v \iff f(u) = f(v)$.

\section{Bikernels into annotated problems}\label{sec:bikernels}

We show in the following that a range of problems over hereditary BE-classes
admit linear bikernels in the same class (see the full version for \RRD and \RTD).
The target problem in all three cases is a suitable annotated version of the
original problem, which we define just ahead of each proof.

%

\begin{problem}[k]{Annotated $(r, c)$-Domination}
  \Input & A graph~$G$, a set~$L \subseteq V(G)$ and an integer~$k$. \\
  \Prob  & Is there a set~$D \subseteq V(G)$ of size at most~$k$ such that
           $|N^r[v] \cap D| \geq c$ for all $v \in L$?
\end{problem}

\begin{theorem}\label{thm:domination-bikernel}
  \RCDS over a hereditary BE-class $\mathcal G$
  admits a linear bikernel into \ARCDS over the same class~$\mathcal G$.
  Moreover, the resulting graph is an $(r,c)$-projection kernel of the original graph.
\end{theorem}
\begin{proof}
  Let~$(G,k)$ be an input where $G$ is taken from a BE class.
  As a first step, we deal with the case~$\ds^\mu_r(G)$ large by computing an $(r,\mu)$-dominating set
  using the algorithm from Theorem~\ref{thm:cdomapprox}. If it returns
  a solution larger than $\ccdom k$, we conclude
  that~$\ds^\mu_r(G) > k$ in which case we return a trivial no-instance.

  Otherwise we now show that \RCDS admits a linear constraint core and then show
  how to construct a BE-kernel from that core.

  \begin{claim}
    \RCDS has a linear constraint core in BE classes.
  \end{claim}
  \begin{proof}
    Let~$L \subseteq V(G)$ be a constraint core of $G$ with $|L| \geq
    \clily[c,2r,r,2] \ds_r^c(G)$. By Corollary~\ref{cor:lily}, we can find in
    polynomial time a uniform water lily $(R,C)$, $C \subseteq L$, $|R| \leq
    \clily$ of depth~$r$, radius~$2r$, adhesion~$c$ and ratio~$2$. Let $a \in C$
    be an arbitrary centre, we claim that $L\setminus \{a\}$ is still a
    constraint core, that is, every set that $(r,c)$-dominates $L \setminus
    \{a\}$ will also $(r,c)$-dominate~$a$.

    To that end, let $D$ be a minimum \rcds and define $D' := D \setminus
    N^r_{G-R}[C]$. If~$D'$ $(r,c)$-dominates any part of $C$, it dominates all of~$C$
    (and therefore~$a$) as $(R,C)$ is uniform. Thus assume that $D'$ does not
    $(r,c)$-dominate $C$. Consider the case where a set~$S \subseteq D \cap N^r_{G-R}[C]$
    exists such that every vertex in~$S$ dominates more than one vertex in~$C$.
    If~$|S| \geq c$ then~$S$ alone already $(r,c)$-dominates all of~$C$ and thus
    in particular~$a$. In all remaining cases, every set~$N^{r}_{G-R}[a']$, $a' \in C$ must
    contain at least one vertex from~$D$ and we conclude that $|D\setminus D'|
    \geq |C| \geq 2|R|$. Let $\tilde D := D' \cup R$, we claim that $\tilde D$
    is an \rcds of~$G$. Simply note that the only vertices that are not
    $(r,c)$-dominated by~$D'$ lie inside~$N^{2r}_{G-R}[C]$---but this is precisely
    the set that is $(r,c)$-dominated by~$R$. We arrive at a contradiction
    since
    \[
      |D| = |D\setminus D'| + |D'| \geq 2|R| + |D'| > |R| + |D'| \geq |\tilde D|
    \]
    and we assumed~$D$ to be minimum.
    Thus~$L \setminus \{a\}$ is a constraint core for \RCDS in $G$. We iterate
    this procedure until~$|L| < \clily[c,2r,r,2] \ds^c_r(G)$ and end up
    with a linear constraint core. 
  \end{proof}

  \noindent
  In the following, let~$L \subseteq V(G)$ be a constraint core
  for~$(G,k)$ with $|L| \leq \clily \ds^c_r(G)$ and let~$O = V(G) \setminus L$.
  If $|L| > \clily k$, we can conclude that $k > \ds^c_r(G)$ and output
  a trivial no-instance, thus assume from now on that $|L| \leq \clily k$.

  We apply Lemma~\ref{lemma:projkernel} with~$X = L$ and~$r$, $c$ as here
  to obtain a projection kernel~$\hat G$ with~$|\hat G| \leq \ctotalclos |L| = O(k)$
  which a) preserves~$\leq r$-neighbourhoods in~$L$ and b) realizes every
  $r$-projection onto~$L$ that is realized~$p$ times in~$G$ at least
  $\min\{c,p\}$ times. We claim that $(G,k)$ is equivalent to the
  annotated instance $(\hat G,L,k)$.

  Assume that~$D$ is an \rcds of~$G$, clearly it is also a solution to the
  annotated instance $(G, L, k)$. Partition $D$ into $D_L = D \cap L$ and
  $D_O = D \setminus L$. Consider $x \in D_O$ and note that $|[x] \cap D_O| < c$
  for the $r$-neighbourhood class~$[x] \in O / \sim^r_L$ since otherwise
  we could remove a vertex from~$[x] \cap D_O$ from~$D$ and still $(r,c)$-dominate
  all of~$L$. With this observation, construct the set~$\hat D_O$ as follows:
  for every vertex~$x \in D_O$ we include $|[x] \cap D_O|$ vertices from~$O \cap V(\hat G)$
  in~$\hat D_O$, by property b) of the projection kernel~$\hat G$ we know that
  at least~$c$ such vertices are available. Then the set~$\hat D := D_L \cup \hat D_O$
  $(r,c)$-dominates all of~$L$ in $\hat G$, by property a) of~$\hat G$, and we
  are done.
  In the other direction, let~$\hat D$ be an $(r,c)$-dominator of~$L$ in $\hat G$.
  By property a) and b) of~$\hat G$ the set~$\hat D$ therefore also $(r,c)$-dominates
  $L$ in~$G$, and since~$L$ is a constraint core of~$G$ it then $(r,c)$-dominates all
  of~$G$.
  We conclude that $(\hat G,L,k)$ is equivalent to~$(G,k)$ and~$|\hat G| = O(k)$.
\end{proof}

\begin{problem}[k]{Annotated Total $r$-Domination}
  \Input & A graph~$G$, a set~$L \subseteq V(G)$ and an integer~$k$. \\
  \Prob  & Is there sets~$D$, $|D| \leq k$ such that for every vertex~$v \in L$
           $|(N^r[v]\setminus\{v\}) \cap D| \geq 1$?
\end{problem}

\begin{theorem}\label{thm:total-domination-bikernel}
  \RTD over a hereditary BE-class $\mathcal G$
  admits a linear bikernel into \ARTD
  over the same class~$\mathcal G$.
  Moreover, the resulting graph is a $(r,1)$-projection kernel of the original graph.
\end{theorem}
\begin{proof}
  Every $r$-total dominating set is in particular an $r$-dominating set
  and, on the other hand, we can turn an $r$-dominating set~$D$ into an $r$-dominating
  set of size at most~$2|D|$ by including at most one neighbour for each vertex in~$D$.

  Hence, given an input~$(G,k)$ to \RTD with $G$ taken
  from a BE class, we verify as a first step that~$\ds_r(G)$ is not too large
  by computing an $r$-dominating set using the algorithm described in Theorem~\ref{thm:cdomapprox}.
  If the  algorithm returns a solution larger than $2\cdvorak k$, we conclude
  that~$\ds_r(G) > 2k$ and therefore $G$ does not have a $r$-total dominating set
  of size~$k$. In this case we output a trivial no-instance, thus assume for
  the remainder that $\ds_r(G) \leq 2\cdvorak k$. Define
  $\clily := \clily[1,2r,r,4]$ in the following.

  \begin{claim}
    \RTD has a linear constraint core in BE classes.
  \end{claim}
  \begin{proof}
    Let~$L \subseteq V(G)$ be constraint core of $G$ with $|L| \geq
    \clily \ds_r(G)$. By Corollary~\ref{cor:lily}, we can find in
    polynomial time a uniform water lily $(R,C)$, $C \subseteq L$, $|R| \leq
    \clily$ of depth~$r$, radius~$2r$, adhesion~$c$ and ratio~$3$. Let $a \in C$
    be an arbitrary centre, we claim that $L\setminus \{a\}$ is still a
    constraint core, that is, every set that totally $r$-dominates $L \setminus
    \{a\}$ will also totally $r$-dominate~$a$.

    To that end, let $D$ be a minimal total $r$-dominating set and define $D' := D \setminus
    N^r_{G-R}[C]$. If~$D'$ totally $r$-dominates any part of $C$, it dominates all of~$C$
    (and therefore~$a$) as $(R,C)$ is uniform. Similarly, if there exists~$u \in D \cap
    N^r_{G-R}[C]$ such that~$u$ dominates at least two centres, then by uniformity it
    already dominates all of~$C$ and in particular~$a$. In all other cases,
    every set~$N^{r}_{G-R}[a']$, $a' \in C$ must
    contain at least one vertex from~$D$ and we conclude that $|D\setminus D'|
    \geq |C| \geq 3|R|$. Let $\tilde D$ consist of~$D'$, $R$ and up to~$|R|$ arbitrary
    neighbours~$R'$ of~$R$. We claim that $\tilde D$
    is a total $r$-dominating set of~$G$.

    Note that the only vertices that are not
    $r$-dominated by~$D'$ lie inside~$N^{2r}_{G-R}[C]$---but this is precisely
    the set that is $r$-dominated by~$R$. The vertices in~$R$ are dominated by
    $R'$ and vice-versa, we conclude that~$\tilde D$ is indeed totally $r$-dominating.
    We arrive at a contradiction since
    \[
      |D| = |D\setminus D'| + |D'| \geq 3|R| + |D'| > 2|R| + |D'| \geq |\tilde D|
    \]
    and we assumed~$D$ to be minimal.

    Thus~$L \setminus \{a\}$ is a constraint core for
    \textsc{Total $r$-Domination} in $G$. We can iterate this procedure
    until~$|L| < \clily \ds^c_r(G)$ and therefore end up
    with a linear constraint core.
  \end{proof}

  \noindent
  In the following, let~$L \subseteq V(G)$ be a constraint core
  for~$(G,k)$ with $|L| \leq \clily \ds_r(G) \leq 2 \clily \cdvorak k$.

  We apply Lemma~\ref{lemma:projkernel} with~$X = L$, $c=1$, and $r$ as here
  to obtain a projection kernel~$\hat G$ with~$|\hat G| \leq \ctotalclos |L| = O(k)$.
  The proof that $(G,k)$ is equivalent to the annotated instance $(\hat G,L,k)$
  is almost identical to the proof in Theorem~\ref{thm:domination-bikernel} and
  we omit it here.
\end{proof}

\begin{problem}[k]{Annotated $r$-Roman Domination}
  \Input & A graph~$G$, a set~$L \subseteq V(G)$ and an integer~$k$. \\
  \Prob  & Is there sets~$D_1, D_2 \subseteq V(G)$ with $|D_1| + 2|D_2| \leq k$
           such that $D_2$ $r$-dominates all of~$L \setminus D_1$?
\end{problem}

\begin{theorem}\label{thm:roman-domination-bikernel}
  \RRD over a hereditary BE-class $\mathcal G$
  admits a linear bikernel into \ARRD
  over the same class~$\mathcal G$.
  Moreover, the resulting graph is a $(r,1)$-projection kernel of the original graph.
\end{theorem}
\begin{proof}
  Let~$(G,k)$ be an input to \RRD where $G$ is taken
  from a BE class. As a first step, we verify
  that~$\ds_r(G)$ is not too large by computing an $r$-dominating set
  using the algorithm described in Theorem~\ref{thm:cdomapprox}. If the
  algorithm returns a solution larger than $\cdvorak k$, we conclude
  that~$\ds_r(G) > k$, since~$(G,k)$ cannot be $r$-dominated by $k$
  vertices it in particular cannot be $r$-Roman dominated with that
  budget. In this case we output a trivial no-instance, thus assume for
  the remainder that $\ds_r(G) \leq \cdvorak k$. Define
  $\clily := \clily[1,2r,r,3]$ in the following.

  \begin{claim}
    \RRD has a linear constraint core in BE classes.
  \end{claim}
  \begin{proof}
    Let~$L \subseteq V(G)$ be constraint core of $G$ with $|L| \geq
    \clily \ds_r(G)$. By Corollary~\ref{cor:lily}, we can find in
    polynomial time a uniform water lily $(R,C)$, $C \subseteq L$, $|R| \leq
    \clily$ of depth~$r$, radius~$2r$, adhesion~$c$ and ratio~$3$. Let $a \in C$
    be an arbitrary centre, we claim that $L\setminus \{a\}$ is still a
    constraint core, that is, every set that $r$-Roman-dominates $L \setminus
    \{a\}$ will also $r$-Roman-dominate~$a$.

    To that end, let $D_1,D_2$ be a $r$-Roman dominating set of minimal
    cost ($|D_1|+2|D_2|$). If~$|\vol^r_R(a) \cap D_2| \geq 2$ this set already
    dominates $a$ and there is nothing to prove, so assume otherwise. Then
    every centre $a' \in C$ must either be contained in $D_1$ or have at least
    one $D_2$-vertex in its pad, \eg $|N^r_{G-R}[a'] \cap D_2| \geq 1$.
    Therefore the total cost of $D_1,D_2$ when restricted to
    $N^r_{G-R}[C]$ is at least $|C| \geq 3|R|$.

    Construct the set $D'_1$ from~$D_1$ by removing
    all vertices in~$N^r_{G-R}[C]$ and construct the set~$D'_2$ from
    $D_2$ by removing all vertices in~$N^r_{G-R}[C]$ and adding all of~$R$.
    Since~$R$ $r$-dominates all of~$N^{2r}_{G-R}[C]$ and all these vertices are
    in $D'_2$, we can conclude that $D'_1,D'_2$ is indeed an $r$-Roman dominating
    set. By our above observation, the cost of $D'_1,D'_2$ is at least $|R|$ smaller
    than the cost of $D_1,D_2$, contradiction minimality.

    Thus~$L \setminus \{a\}$ is a constraint core for \RCDS in $G$. We can iterate
    this procedure until~$|L| < \clily \ds_r(G)$ and therefore end up
    with a linear constraint core.
  \end{proof}

  \noindent
  In the following, let~$L \subseteq V(G)$ be a constraint core
  for~$(G,k)$ with $|L| \leq \clily \ds_r(G) \leq \clily \cdvorak k$.

  We apply Lemma~\ref{lemma:projkernel} with~$X = L$, $c=1$, and $r$ as here
  to obtain a projection kernel~$\hat G$ with~$|\hat G| \leq \ctotalclos |L| = O(k)$.
  The proof that $(G,k)$ is equivalent to the annotated instance $(\hat G,L,k)$
  is almost identical to the proof in Theorem~\ref{thm:domination-bikernel} and
  we omit it here.
\end{proof}

%
%

\begin{problem}[k]{Annotated $(r,c)$-Scattered Set}
  \Input & A graph~$G$, a set~$U \subseteq V(G)$ and an integer~$k$. \\
  \Prob  & Is there a set~$I \subseteq U$ of size at least~$k$ such that
           $|N^r[v] \cap I| \leq c$ for all $v \in V(G)$?
\end{problem}

\noindent
The following proof makes use of the pad equivalence $\simuni$ defined in
Section~\ref{sec:lily}: recall two centres~$u,v$ of a water lily~$(R,C)$
satisfy $u \simuni v$ if they have the same
projection-types onto~$R$ at the same distance (for distances smaller than the
lily's depth) inside their respective pads.

\begin{theorem}\label{thm:scattered-bikernel}
  \RCSC over a hereditary BE-class $\mathcal G$
  admits a linear bikernel into \ARCSC over the same class~$\mathcal G$.
  Moreover, the resulting graph is an $(r,c)$-projection kernel of the original graph.
\end{theorem}
\begin{proof}
  Let $(G,k)$ be an instance of \RCSC where $G$ is taken from a BE class.
  As a first step, we deal with the case
  that~$\scat^c_r(G)$ is large. We compute an $\cdvorak$-approximate
  $r$-dominating set~$D$ using Theorem~\ref{thm:dvorak-ds}. If
  $|D| > \cdvorak \wcolnum_{2r}^2(G) \cdot k$, we conclude by
  Theorems~\ref{thm:dvorak-duality} and~\ref{thm:dvorak-is} that
  $\scat^c_r(G) \geq \scat_r(G) > k$
  and we output a trivial yes-instance. Otherwise, assume
  $|D| \leq \cdvorak \wcolnum_{2r}^2(G) \cdot k$ and
  define $\clily := \clily[1,2r,r,2,\nu]$. We first show
  that \RCSC admits a linear solution core.
  \begin{claim}
    \RCSC has a linear solution core in BE classes.
  \end{claim}
  \begin{proof}
    Let~$U \subseteq V(G)$ be solution core of $G$ with $|U| \geq
    \clily \ds_r(G)$. Using Corollary~\ref{cor:lily}, we find in
    polynomial time a $\nu$-uniform water lily $(R,C)$, $C \subseteq U$, $|R| \leq
    \clily$ of depth~$r$, radius~$2r$, adhesion~$1$ and ratio~$2$. Let $a \in C$
    be an arbitrary centre, we claim that $U\setminus \{a\}$ is still a solution
    core, \ie there exists an optimal \rcsc that does not contain~$a$.

    To that end, let $I$ be a minimum \rcsc and assume $a \in I$. We claim
    that there exists an \rcsc $I'$ of the same size which excludes~$a$.
    First observe that every vertex that lives in a pad~$N^{2r}[a']$, $a' \in
    C$, has at least~$c$ neighbours in~$R$ at distance~$\leq r$. Therefore
    $|N^{2r}_{G-R}[C] \cap I| \leq |R|$ as otherwise we would find a vertex
    in~$R$ whose $r$-neighbourhood contains more than~$c$ vertices of~$I$.
    Since~$|C| \geq 2|R|$ there are at least~$|R|$ centres~$C' \subseteq C$ such that
    their pads $N^{2r}_{G-R}[C']$ do not intersect~$I$.
    Since~$(R,C)$ is uniform and~$a \in I$, we know that $|N^r[a'] \cap
    I| = |N^r[a] \cap I| < c$ for every centre~$a \in C$.

    Take $a' \in C'$ and let~$I' := I\setminus\{a\} \cup\{a'\}$. To see that
    $I'$ is $(r,c)$-scattered, consider any vertex $u' \in N^r[a']$ (note that
    vertices at distance~$> r$ from~$a'$ are not affected by the exchange of~$a$
    by $a'$). By $\nu$-uniformity, there exists a vertex~$u \in N^r[a]$ with
    $\pi^r_R[u] = \pi^r_R[u']$. In particular, $\proj^r_R(u) \cup \shad^r_R(u) =
    \proj^r_R(u') \cup \shad^r_R(u')$; therefore $(N^r[u] \cap I) \setminus
    \{a\} = (N^r[u'] \cap I') \setminus \{a'\}$ and we conclude that $|N^r[a']
    \cap I'| \leq c$.
    It follows that $U\setminus \{a\}$ is a solution core. We iterate the
    above procedure until $|U| \leq \clily \ds^c_r(G)$ and end up with
    a linear solution core.
  \end{proof}

  \noindent
  In the following, let~$U \subseteq V(G)$ be a solution core
  for~$(G,k)$ with $|U| \leq \clily \ds_r(G) \leq \clily |D|
  = O(k)$.

  We apply Lemma~\ref{lemma:projkernel} with~$X = U$ and~$r$, $c$ as here
  to obtain a projection kernel~$\hat G$ with~$|\hat G| \leq \ctotalclos |U| = O(k)$
  that a) preserves~$\leq r$-neighbourhoods in~$U$ and b) realizes every
  $r$-projection onto~$U$ that is realized~$p$ times in~$G$ at least
  $\min\{c,p\}$ times. Since distances in~$\hat G[U]$ are as in $G[U]$,
  it is easy to see that any set~$I \subseteq U$ is $(r,c)$-scattered in~$\hat G$
  iff it is $(r,c)$-scattered in~$G$. Since~$U$ is further a solution core
  for~$G$, we conclude that $(G,k)$ is equivalent to the
  annotated instance $(\hat G,U,k)$.
\end{proof}


\noindent
We show that
\textsc{$(r,[\lambda, \mu])$-Domination} admits a linear bikernel
into the following annotated problem:

\begin{problem}[k]{Annotated $(r,[\lambda, \mu])$-Domination}
  \Input & A graph~$G$, sets~$L, U \subseteq V(G)$ and an integer~$k$. \\
  \Prob  & Is there a set~$D \subseteq U$ of size at most~$k$ such that
           $|N^r[v] \cap D| \geq \lambda$ for all $v \in L$ and
           $|N^r[v] \cap D| \leq \mu$ for all $v \in V(G)$?
\end{problem}

\noindent
We note that the construction in the following proof results in a bikernel
$(\hat G, L, U, k)$ with $L \subseteq U$, the construction can also be
easily be modified to ensure that $L = U$.

\begin{theorem}\label{thm:lambdamu-domination-bikernel}
  \Problem{$(r,[\lambda, \mu])$-Domination} over a hereditary BE-class~$\mathcal G$ admits
  a linear bikernel into \Problem{Annotated $(r,[\lambda, \mu])$-Domination}
  over the same class~$\mathcal G$.
  Moreover, the resulting graph is an $(r,c)$-projection kernel of the original graph.
\end{theorem}
\begin{proof}
  Since the cases where either~$\mu = \infty$ or $\lambda = 0$ are
  equivalent to \RCDS or \RCSC and thus covered by Theorems~\ref{thm:domination-bikernel}
  and~\ref{thm:scattered-bikernel}, we here only consider the case of
  $\lambda \neq 0$ and $\mu \neq \infty$.

  Note that any solution to the problem is in particular an
  $(r,\mu)$-dominating set. As a first step, we therefore deal with the case
  that~$\ds^\mu_r(G)$ is too large by computing an $(r,\mu)$-dominating set
  using the algorithm described in Theorem~\ref{thm:cdomapprox}. If the
  algorithm returns a solution larger than $\ccdom k$, we conclude
  that~$\ds^\mu_r(G) > k$ and therefore that~$(G,k)$ must be a no-instance;
  in which case we output a trivial no-instance.
  Otherwise, let~$\hat D$ be the resulting $(r,c)$-dominating set.

  Let~$(G,L,U,k)$ be an instance of \Problem{Annotated
  $(r,[\lambda,\mu])$-Domination} with~$L = U = V(G)$. Clearly, $(G,L,U,k)$ is
  equivalent to~$(G,k)$. In the following, we gradually reduce the size of~$L$
  and~$U$ while maintaining this equivalence. To that end, we will use the pad
  signature~$\nu$ which is to be understood to take the `vertex labels' $L$,
  $U$ into account.

  Assume that $|L| > (\clily+1) |\hat D|$ with~$\clily :=
  \clily[r,2r,\mu+1,\nu]$. Then, using~$\hat D$ in the construction used in
  the proof of Lemma~\ref{lemma:lily}, we find a $\nu$-uniform water
  lily~$(R,C)$ with~$C \subseteq L\setminus \hat D$ of depth~$r$,
  radius~$2r$ and ratio $(\mu+1)$.
  \begin{claim}
    Let $a \in C$.
    Then the instances~$(G,L,U,k)$ and $(G,L\setminus\{a\},U,k)$ are equivalent.
  \end{claim}
  \begin{proof}
    Any solution for~$(G,L,U,k)$ is also a solution
    to~$(G,L\setminus\{a'\},U,k)$, therefore we only have to show the opposite
    direction.

    Let~$D$ be a solution for~$(G,L\setminus\{a\},U,k)$. Since~$R
    \subseteq L \cap U$, the set $D$ can intersect at most~$\mu |R|$ pads or
    otherwise we would violate an upper constraint for at least one of the
    vertices in~$R$. It follows that at least $|R|$ pads of $(R,C)$
    cannot contain any vertex of~$D$; let the centres of these pads be~$C'
    \subseteq C$. Choose~$a' \in C'$ distinct from~$a$ (since~$|C'|
    \geq |R| \geq \lambda > 1$ such a vertex exists). Note that $a' \in L$, therefore
    $|N^r[a'] \cap D| \geq \lambda$. But since~$N^r_{G-R}[a'] \cap D = \emptyset$,
    these solution vertices must lie in~$\vol^r_R(a')$. Now simply
    observe that, by uniformity of~$(R,C)$,
    $\vol^r_R(a) = \vol^r_R(a')$ and therefore $|N^r[a'] \cap D| \geq |\vol^r_R(a) \cap D| \geq \lambda$.
    Accordingly, $D$ is also a solution for~$(G,L,U,k)$.
  \end{proof}

  \noindent
  We repeat the above procedure until $|L\setminus \hat D| \leq \clily k$. Now
  assume that $|U \setminus (L \cup \hat D)| > \clily k$ and let~$(R,C)$ be a
  $\nu$-uniform water lily with~$C
  \subseteq U \setminus (L \cup \hat D)$ of depth~$r$, radius~$2r$ and
  ratio $(\mu+1) |R|$.


  \begin{claim}
    Let $a \in C$.
    Then the instances~$(G,L,U,k)$ and $(G,L,U\setminus\{a\},k)$ are equivalent.
  \end{claim}
  \begin{proof}
    By construction of~$(R,C)$, every vertex~$x \in N^{2r}_{G-R}[C]$
    is $(r,\mu)$-dominated by~$R \cap \hat D$. Importantly, $R \cap \hat D
    \subseteq R \cap U$, therefore any solution~$D$ of~$(G,L,U,k)$ can
    intersect~$N^r[R]$ in at most~$\mu |R|$ vertices. In particular,
    at most~$\mu |R|$ pads of~$(R,C)$ can contain vertices of~$D$,
    let us call the centres of these empty pads~$C' \subseteq C$.

    If~$a \not \in D$, clearly~$D$ is a solution of~$(G,L,U\setminus\{a\},k)$
    and there is nothing to prove. Assume therefore that~$a \in D$. Let $a' \in
    C'$ be an arbitrary centre of an empty pad. We claim that $D' := D \setminus
    \{a\} \cup \{a'\}$ is a solution to $(G,L,U\setminus\{a\},k)$. To that
    end, consider any vertex~$x \in N^r[a] \cup N^r[a']$, we will show
    that~$D'$ fulfils any constraints associated with~$x$.

    \begin{case} $x \in N_{G-R}^r[a]$. \newline
      By $\nu$-uniformity, there exists a vertex $x' \in N_{G-R}^r[a']$ such
      that $\vol^r_{G-R}(x) = \vol^r_{G-R}(x')$ and $x'$ is contained in $L$ ($U$)
      iff $x$ is contained in $L$ ($U$).
      For the special case that $x = a$ we let $x' = a'$.

      Assume~$x \in L$, then~$x' \in L$ and accordingly~$|N^r[x'] \cap D| \geq \lambda$.
      Since~$N^{2r}_{G-R}[a'] \cap D = \emptyset$, we have
      that
      \[
        N^r[x'] \cap D = \vol^r_R(x') \cap D = \vol^r_R(x') \cap D'
        = \vol^r_R(x) \cap D',
      \]
     therefore $|N^r[x] \cap D'| = |N^r[x'] \cap D| \geq \lambda$ and the
     lower-bound constraint for~$x$ is satisfied by~$D'$.

     If~$x \in R$, simply note that~$|N^r[x] \cap D'| \leq |N^r[x] \cap D|
     \leq \mu$, hence the upper-bound constraint for~$x$ is satisfied by~$D'$.
    \end{case}

    \begin{case} $x \in N_{G-R}^r[a']$
      Again, by $\nu$-uniformity, there exists a vertex $\hat x \in
      N_{G-R}^r[a]$ such that $\vol^r_{G-R}(x) = \vol^r_{G-R}(\hat x)$ and
      $\hat x$ is contained in $L$ ($U$) iff $x$ is contained in $L$ ($U$).
      For the special case that $x = a'$ we let $\hat x = a$.

      If~$x \in L$, simply note that $|N^r[x] \cap D'| \geq |N^r[x] \cap D|
      \geq \lambda$, hence the lower-bound constraint for~$x$ is
      satisfied by~$D'$.

      Assume~$x \in R$. Then~$\hat x \in R$ and accordingly
      $|N^r[\hat x] \cap D| \leq \mu$. More specifically, since~$a
      \in N^r[\hat x] \cap D$, we know that
      $|\vol^r_R[\hat x] \cap D| \leq \mu-1$. Because
      \[
        N^r[x] \cap D' = (\vol^r_R[x] \cap D') \cup \{a'\}
        =  (\vol^r_R[\hat x] \cap D) \cup \{a'\}
      \]
      we conclude that~$|N^r[x] \cap D'| \leq \mu$ and the upper-bound
      constraint for~$x$ is satisfied by~$D'$.
    \end{case}

    \begin{case} $x \in \vol^r_R(a) = \vol^r_R(C)$. \quad
        Simply note that by
        uniformity $|N^r[x] \cap D| = |N^r[x] \cap D'|$ and
        therefore~$D'$ satisfies all constraints for~$x$.
    \end{case}

    \noindent
    Therefore~$D'$ is indeed a solution for~$(G,L,U\setminus \{a\},k)$
    of equal size and we conclude that the instances
    $(G,L,U,k)$ and~$(G,L,U\setminus \{a\},k)$ are equivalent, as claimed.
  \end{proof}

  \noindent
  We repeat the above procedure until $|U \setminus (L \cup \hat D)| \leq
  \clily k$ and end up with an instance~$(G,L,U,k)$ which is equivalent to our
  initial instance~$(G,k)$ and further satisfies $|L| \leq \clily k$ and $|U|
  \leq |L| + |\hat D| + |U \setminus (L \cup \hat D)| \leq (2 \clily + \ccdom)
  k$.

  Finally, let us construct the bikernel from this annotated instance.
  Note that, by construction,
  $L \subseteq U$.
  Let~$\hat U$ be the shortest-path closure of~$U$
  in~$G$ as per Lemma~\ref{lemma:pathclos}, then~$|\hat U| \leq
  \cpathclos |U|$ and~$\hat G := G[\hat U]$ preserves all distances up to
  length~$r$ between vertices in~$U$. In particular,
  $N^r_{\hat G}[v] \cap U = N^r_G[v] \cap U$. Since the annotated instance
  asks for solutions contained entirely in $U$ and $L \subseteq U$, we
  conclude that the instance~$(G,L,U,k)$ and $(\hat G, L, U, k)$ are
  equivalent, therefore the latter is also equivalent to~$(G,k)$ which
  finally proves the claim.
\end{proof}


\section{From bikernels to BE-kernels}\label{sec:be-kernels}

If we sacrifice the constraint to construct a (bi)kernel that is contained
in the same hereditary graph class, we are able to construct BE-kernels
by reducing from the annotated problem back into the original problems.
In the following constructions, we usually tried to minimize the increase
in the parameter~$k$, not the increase of the expansion characteristics of
the class.

\begin{theorem}\label{thm:rcds-kernel}
  \RCDS admits a linear BE-kernel.
\end{theorem}
\begin{proof}
  For an instance~$(G,k)$ of \RCDS, where~$G$ is taken from a BE class, we first construct a bikernel
  $(\hat G,L,k)$ of \ARCDS according to Theorem~\ref{thm:domination-bikernel}.
  Recall that~$\hat G$ is an $(r,c)$-projection kernel of~$(G,L)$.

  First consider $r \geq 2$. We construct~$G'$ from $\hat G$ by adding
  new vertices $a_1,\ldots,a_c, b_1, b_2, b_3$ to the graph.
  We connect every~$a_i$, $1 \leq i \leq c$ to both~$b_1$ and~$b_2$;  then
  connect~$b_1$ to every vertex in~$O := V(\hat G)\setminus L$  via a path of length
  $r-1$ and connect~$b_2$ to $b_3$ by such a path as well.

  From the construction it is clear that~$G'$ has
  size~$O(k)$, we are left with proving that the two instances~$(G,k)$
  and~$(G', k+c)$ are equivalent.

  Assume that~$D'$ is a minimum \rcds for~$G'$ of size~$\leq k+c$.
  By a simple exchange argument, we can assume that~$D'$
  contains all vertices~$a_i$ in order to $(r,c)$-dominate~$b_3$. These vertices
  already $(r,c)$-dominate all of~$O$ and the paths leading from~$b_1$
  to~$O$. As such, we can assume that an optimal solution~$D'$ does not
  contain internal vertices of those paths (otherwise we might as
  well exchange an internal vertex for the path's endpoint in~$O$).
  Then the set~$\hat D := D' \setminus
  \{a_1,\ldots a_c\}$ has size at most~$k$ and $(r,c)$-dominates all of~$L$;
  thus~$\hat D$ in particular is a solution to~$(\hat G, L, k)$.

  In the other direction, assume that~$\hat D$ is a minimum solution
  for~$(\hat G,L,k)$, that is, $\hat D$ $(r,c)$-dominates $L$ in $\hat G$.
  Let~$D' := \hat D \cup \{a_1,\ldots,a_c\}$, it is easy to see that~$D'$
  $(r,c)$-dominates $G'$ and has size~$|D'| = |D| + c$.

  For~$r=1$ we modify the construction as follows: we add vertices~$a_1,\ldots,a_c,b$
  and connect all~$a_i$ to~$O \cup \{b\}$. The argument for why the resulting
  instance is equivalent is very similar to the case~$r\geq 2$ and we omit it
  here.

  We conclude that~$(\hat G,L,k)$ and~$(G',k+c)$ are indeed equivalent,
  and thus also to~$(G,k)$. It
  is only left to show that the construction of~$G'$ increased the
  expansion characteristics by some arbitrary function independent of~$|G|$.
  Simply note that we can construct~$G'$ from~$G$ by adding~$c+3$ apex-vertices
  (which increases the expansion characteristics only by an additive constant)
  and then remove or subdivide edges incident to them (which does not increase
  the expansion characteristics).
\end{proof}

\begin{theorem}\label{thm:total-domination-kernel}
  \RTD admits a linear BE-kernel.
\end{theorem}
\begin{proof}
  For an instance~$(G,k)$ of \RTD we first construct a bikernel
  $(\hat G,L,k)$ of \ARTD according to Theorem~\ref{thm:total-domination-bikernel}.
  Recall that~$\hat G$ is an $(r,1)$-projection kernel of~$(G,L)$.

  We construct~$G'$ from $\hat G$ as
  follows: add new vertices $b, a_1, a_2$ to the graph.
  Connect~$b$ to every vertex in~$O := V(\hat G)\setminus L$ and to $a_1$
  via a path of length $r$. Then connect~$a_1$ to~$a_2$ by a path of length~$r$.
  It is is clear that~$G'$ has size~$O(k)$, we are left with proving
  that the two instances~$(G,k)$ and~$(G', k+2)$ are equivalent.

  From the construction it is clear that~$G'$ has
  size~$O(k)$, we are left with proving that the two instances~$(G,k)$
  and~$(G', k+2)$ are equivalent.

  First, assume that~$D'$ is a minimal \rtd for~$G'$. Since the path
  from~$b$ to~$a_2$ needs to contain at least one vertex to dominate the
  path, we can, by a simple exchange argument, assume that this vertex is~$a_1$.
  $D'$ further needs to dominate~$a_1$ itself, again by an exchange argument
  we may assume that~$b \in D'$. We can therefore assume that~$D'$ does not
  contain the paths between~$b$ and~$O$ (excluding the vertices~$O$)
  and the path from~$b$ to~$a_2$ in vertices other than~$b,a_1$.
  Then the set~$\hat D := D' \setminus \{b,a_1\}$ has size~$|D'| - 2$
  and totally $r$-dominates all of~$L$, therefore~$\hat D$ is a solution
  to~$(\hat G, L, k)$.

  In the other direction, assume that~$\hat D$ is a minimal solution for~$(\hat G, L, k)$,
  that is, $\hat D$ totally $r$-dominates $L$ in $\hat G$.
  Let~$D' := \hat D \cup \{b,a_1\}$. Then~$D'$ totally $r$-dominates $G'$ and
  has size~$|D|+2$.

  We conclude that~$(\hat G,L,k)$ and~$(G',k+2)$ are indeed equivalent,
  and the latter is also equivalent to~$(G,k)$. The argument the increase of
  the expansion characteristic of~$G'$ is similar to before, we omit it here.
\end{proof}

\begin{theorem}\label{thm:roman-domination-kernel}
  \RRD admits a linear BE-kernel.
\end{theorem}
\begin{proof}
  For an instance~$(G,k)$ of \RRD we first construct a
  bikernel $(\hat G,L,k)$ of \ARRD according
  to Theorem~\ref{thm:roman-domination-bikernel}.
  Recall that~$\hat G$ is an $(r,1)$-projection kernel of~$(G,L)$.

  We construct~$G'$ from $\hat G$ as
  follows: add new vertices $b, a_1, a_2, a_3$ to the graph.
  Connect~$b$ to every vertex in~$O := V(\hat G)\setminus L$ and to $\{a_1,a_2,a_3\}$
  via a path of length $r$. It is is clear that~$G'$ has
  size~$O(k)$, we are left with proving that the two instances~$(G,k)$
  and~$(G', k+2)$ are equivalent.

  First, assume that~$D'_1,D'_2$ is a minimal $r$-Roman dominating set for~$G'$ of size at
  most~$k+2$. By a simple exchange argument, we can assume that~$b \in D'_2$
  in order to $r$-Roman-dominate $a_1$, $a_2$, and~$a_3$ (including all three
  vertices in $D'_1$ would be more expensive). Now~$b$ already $r$-Roman-dominates
  all of~$O$ as well as the paths added during the construction, we can therefore
  assume that~$D'_1$ is entirely contained in~$V(G)$. Therefore
  the sets $D'_1,D'_2 \setminus \{b\}$ $r$-Roman-dominates $L$ in~$G$ at a cost
  of~$|D'_1|+2|D'_2|-2$.

  In the other direction, assume that~$\hat D_1, \hat D_2$ is a minimal-cost solution
  for~$(\hat G,L,k)$, that is, $\hat D_1, \hat D_2$ $r$-Roman-dominates $L$ in $\hat G$.
  Partition both set $\hat D_i$ for~$i \in [1,2]$ into sets
  $\hat D_{i,O} = \hat D_i \setminus L$ and $\hat D_{i,L} = \hat D_i \cap L$.
  Then construct~$D'_{i,O}$ as follows:  for every equivalence class~$[u] \in \hat D_{i,O} / \sim^r_L$,
  include a vertex of~$[u] \cap O$ in~$D'_{i,O}$ (here we use that~$\hat G$
  is am $(r,1)$-projection kernel of~$(G,L)$. Since we picked the same
  projection-classes as in~$\hat D_{1,O}$, $\hat D_{2,O}$, we conclude that
  the sets~$D'_{1,O} \cup \hat D_{1,L}$, $D'_{2,O} \cup \hat D_{2,L}$
  $r$-Roman-dominate the core~$L$. Therefore, the sets
  \[
    D'_1 := D'_{1,O} \cup \hat D_{1,L}, \quad\quad
    D'_2 := D'_{2,O} \cup \hat D_{2,L} \cup \{ b \}
  \]
  $r$-Roman-dominate all of~$G'$ at cost of~$|\hat D_1| + 2|\hat D_2|+2$.

  We conclude that~$(\hat G,L,k)$ and~$(G',k+2)$ are indeed equivalent,
  and the latter is also equivalent to~$(G,k)$. To see that the
  expansion characteristics only increase by a function that is independent
  of~$|G|$,  simply note that we can construct~$G'$ by adding one apex-vertex
  to~$G$ with an additional pendant vertex
  (which increases the expansion characteristics only by an additive constant)
  and then subdivide edges incident to it (which does not increase
  the expansion characteristics).
\end{proof}

\begin{theorem}\label{thm:rcsc-kernel}
  \RCSC admits a linear BE-kernel.
\end{theorem}
\begin{proof}
  Let $(G,k)$ be an input of \RCSC where $G$ is taken from a BE
  class. We first construct the annotated bikernel $(\hat G,U,k)$ according to
  Theorem~\ref{thm:scattered-bikernel} and then construct~$G'$ from~$\hat G$
  by adding vertices~$a_1,a_2, b_1, \ldots, b_c$ and edges~$a_2b_i$ for
  all $1 \leq i \leq c$. We further connect~$a_1$ to all vertices
  in $O := V(\hat G)\setminus U$ via paths of length~$r$ and to~$a_2$ via a path of length~$r-1$
  (for $r=1$ we identify $a_1$ and~$a_2$).
  It is is clear that~$G'$ has size~$O(k)$, we are left to prove that the
  instances~$(\hat G, U, k)$ and~$(G', k+c)$ are equivalent.

  First, consider a maximal $(r,c)$-scattered set~$I'$ in $G'$. Since
  $O \cup \{b_1,\ldots,b_c\} \subset N^r[a_1]$ we may assume,
  by a simple exchange argument, that~$\{b_1,\ldots,b_c\} \subseteq I'$.
  Accordingly, $O \cap I' = \emptyset$ and $I := I' \setminus\{b_1,\ldots,b_c\}$
  is an $(r,c)$-scattered set contained entirely in~$U$. Therefore
  $I$ is $(r,c)$-scattered in~$\hat G$ as well and $|I| = |I'| + c$.

  In the other direction, assume that $\hat I \subseteq U$ is a maximal
  $(r,c)$-scattered set in~$\hat G$. Then~$N^r_{G'}[a_1] \cap I = \emptyset$
  and we can add up to~$c$ vertices from~$N^r[a_1]$
  to~$I$. Since the vertices~$b_i$ all lie at distance~$2r$ from~$O$,
  we conclude that~$I' := I \cup \{b_1,\ldots,b_c\}$ is indeed
  $(r,c)$-scattered in~$G'$ and~$|I'| = |I| + c$.

  We conclude that the instances~$(\hat G,U,k)$ and~$(G',k+c)$ are
  equivalent and hence~$(G,k)$ and~$(G',k+c)$ are as well. The argument
  why the expansion characteristics only increase by a constant are
  similar to the arguments in Theorem~\ref{thm:rcds-kernel}.
\end{proof}

\begin{theorem}
  \textsc{$r$-Perfect Code} admits a linear BE-kernel.
\end{theorem}
\begin{proof}
  Let $(G,k)$ be an input instance of \textsc{$r$-Perfect Code}
  where $G$ is taken from a BE class. Since \textsc{$r$-Perfect Code}
  is equivalent to \textsc{$(r,[1,1])$-Domination}, we proceed by first
  constructing the annotated bikernel $(\hat G, L, U, k)$ according to
  Theorem~\ref{thm:lambdamu-domination-bikernel}. As commented there,
  we can construct the bikernel that $L = U$ which we will assume
  in the following for simplicity.

  Let~$O := V(\hat G)\setminus L$. We construct~$G'$ from~$\hat G$
  by appending a path~$P_u$ of length~$2r$ to every vertex~$u \in O$.
  We claim that~$(\hat G,L,U,k)$ is equivalent to~$(G',k+|O|)$.
  In the following, fix one path~$P_u$ and let~$a_1,\ldots, a_{2r}$ be its
  vertices ordered by their respective distance from~$u$; the arguments we make
  will hold symmetrical for all paths added in the construction.

  First, consider an $r$-perfect code~$D'$ of~$G'$. In order to dominate the
  vertex~$a_{2r}$, it needs to contain a vertex~$a_j \in P_u$ with
  $r \leq j \leq 2r$. Since~$a_j$ will in particular dominate~$a_r$,
  we conclude that~$u \not \in D'$ and, by symmetry, that~$D' \cap O = \emptyset$.
  Then the set~$\hat D := D' \cap V(\hat G)$ is indeed a perfect code for
  $\hat G$ of size~$|D'| - |O|$.

  In the other direction, assume that~$\hat D \subseteq L$ is a perfect
  code for~$L$ in~$\hat G$. Since~$L = U$ is both a solution-
  and a constraint core for~$G$, we know that the set~$\hat D$ is
  a perfect code in~$G$. Because~$\hat G$ is an induced subgraph of~$G$,
  we conclude that~$|N^r[u] \cap \hat D| \leq 1$ for all~$u \in O$.
  Let~$d_u$ be the distance of~$u \in O$ to the closest
  vertex in~$\hat D$ (this distance is, by construction, the same in~$\hat G$
  and~$G'$). We construct~$D'$ from~$\hat D$ as follows: if~$d_u > r$, then
  we add the vertex~$a_r$. Otherwise, note that the vertices~$a_1,\ldots,a_i$
  for~$i = r - d_u$ of~$P_u$ are dominated by~$\hat D$, we therefore
  add the vertex~$a_j$ with $j = 2r-d_u+1$. The resulting set~$D'$ dominates,
  in~$G'$, all vertices in~$O$ that are not dominated by~$\hat D$ and further
  dominates all vertices~$V(G')\setminus V(\hat G)$ precisely once. It follows
  that~$D'$ is a perfect code in~$G'$ of size~$|\hat D| + |O|$.

  We conclude that the instances~$(\hat G,L,U,k)$ and~$(G',k+|O|)$ are
  equivalent and hence~$(G,k)$ and~$(G',k+|O|)$ are as well. The construction
  of~$G'$ from~$\hat G$ only increases the expansion characteristics if the
  original graph class consist of edgeless graphs.
\end{proof}


\section{Multikernels}\label{sec:multikernels}

The following results are applicable to graph BE-classes that are closed
under the addition of pendant vertices, \eg planar graphs,
graphs of bounded genus or graph classes defined by an excluded minor of
minimum degree two. Their proofs are a collection of arguments already made in
detail in Section~\ref{sec:be-kernels}, we will abbreviate those parts here.
In the following, let~$\ds_r^\text{total}(G)$ denotes the total $r$-domination
number and ~$\ds_r^\text{roman}(G)$ the $r$-Roman domination number of~$G$.
We will also write~$\ds_r(G, L)$, $\ds_r^\text{total}(G, L)$, and
$\ds_r^\text{roman}(G,L)$ for the annotate domination numbers (where only
the set~$L \subseteq V(G)$ has to be dominated).

\begin{theorem}
  Let~$\mathcal G$ be a hereditary BE-class that is further closed under
  adding pendant vertices.
  Given a graph~$G \in \mathcal G$ and an integer~$r$ we can
  compute in polynomial time a graph~$G' \in \mathcal G$ and an integer~$c$
  with the following properties:
  \begin{itemize}
    \item $|G'| = O(\ds_r(G)) = O(\ds^{total}_r(G)) = O(\ds^{roman}_r(G))$,
    \item $\ds_r(G') = \ds_r(G) + c$,
    \item $\ds^\text{total}_r(G') = \ds^\text{total}_r(G) + c$, and
    \item $\ds^\text{roman}_r(G') = \ds^\text{roman}_r(G) + 2c$.
  \end{itemize}
\end{theorem}
\begin{proof}
  We apply the constructions from Theorems~\ref{thm:domination-bikernel} (for $c=1$),
  \ref{thm:total-domination-bikernel}, and~\ref{thm:roman-domination-bikernel}
  to find constraint cores~$L_d,L_t$ and~$L_r$ for all three problems.
  Since~$\ds_r^\text{total}$ and $\ds_r^\text{roman}$ lie within
  a factor of two of~$\ds_r$, we conclude that the joint set~$K := L_d \cup L_t \cup L_r$
  is a constraint core for all three problems of size~$|K| = O(\ds_r(G))$.

  Let~$\hat G$ be a $(\mu,1)$-projection kernel of~$(G,K)$ constructed according
  to Lemma~\ref{lemma:projkernel}, recall that~$\hat G$ is an induced subgraph of~$G$
  with~$|\hat G| = O(|K|)$. By the proofs of Theorems~\ref{thm:domination-bikernel},
  \ref{thm:total-domination-bikernel}, and~\ref{thm:roman-domination-bikernel}
  we have that~$\ds_r(\hat G, L) = \ds_r(G)$, $\ds^\text{total}_r(\hat G, L) = \ds^\text{total}_r(G)$,
  and $\ds^\text{roman}_r(\hat G, L) = \ds^\text{roman}_r(G)$.

  Let~$O := V(\hat G)\setminus K$ be the vertices outside the core set~$K$
  and let~$c := 2|O|$.
  Let~$T$ be the tree constructed as follows: create vertices~$b_0, b_1, b_2,
  a_1,\ldots,a_6$. Connect~$b_0$ to~$b_1$ and~$b_1$ to~$b_2$ by paths of
  length~$r$, then connect~$b_1$ to~$a_1,\ldots,a_3$ each by a path of length~$r$
  and $b_2$ to~$a_4,\ldots, a_6$.
  We construct~$G'$ by appending to each vertex~$v \in O$ a copy~$T_v$ of~$T$
  by identifying~$b_0$ with~$v$. It is not difficult to see that any optimal
  \rds and \rtd can, by an exchange argument, be assumed to contain
  the vertices~$b_1$ and~$b_2$ of each tree~$T_v$; and that any
  \textsc{$r$-Roman-dominating set} includes $b_1$ and $b_2$ at a cost of two
  each. We conclude that
  \begin{align*}
    \ds_r(G') &= \ds_r(\hat G, L) + c = \ds_r(G) + c, \\
    \ds^\text{total}_r(G') &= \ds^\text{total}_r(\hat G, L) + c = \ds^\text{total}_r(G) + c, \text{and} \\
    \ds^\text{roman}_r(G') &= \ds^\text{roman}_r(\hat G, L) + 2c = \ds^\text{roman}_r(G) + 2c.
  \end{align*}

\end{proof}

\noindent
Recall that an $r$-scattered set is equivalent to
a $2r$-independent set and in particular that~$\scat_r(G) = \is_{2r}(G)$.

\begin{theorem}
  Let~$\mathcal G$ be a hereditary BE-class that is further closed under
  adding pendant vertices.
  Given a graph~$G \in \mathcal G$ and integers~$\lambda \leq \mu$ we can
  compute in polynomial time a graph~$G' \in \mathcal G$ and
  integers~$c_{\lambda},\ldots,c_{\mu}$ with the following properties:
  \begin{itemize}
    \item $|G'| = O(\ds_{\lambda}(G))$,
    \item for all~$\lambda \leq r \leq \mu$ it holds that
          $\ds_r(G') = \ds_r(G) + c_r$, and
    \item for all~$\lambda \leq r \leq \mu$ it holds that~$\is_{2r}(G') = \is_{2r}(G) + c_r$.
  \end{itemize}
\end{theorem}
\begin{proof}
  We apply the constructions from Theorems~\ref{thm:domination-bikernel}
  and~\ref{thm:scattered-bikernel} for~$r \in [\lambda,\mu]$ to construct
  constraint cores~$L_r$ for \RDS and solution cores~$S_r$
  for~\textsc{$r$-Scattered Set}. Let $L := \bigcup_{\lambda \leq r \leq \mu}
  L_r$ and $S := \bigcup_{\lambda \leq r \leq \mu} S_r$; since~$|L_r| =
  O(\ds_r(G))$ and~$|S_r| = O(\is_r(G))$ and~$\is_{2r}(G) = \Theta(ds_r(G))$
  by Theorem~\ref{thm:dvorak-duality} we conclude that $|L \cup S| =
  O((\mu-\lambda) \ds_\lambda(G)) = O(\ds_\lambda(G))$. Define $K := L \cup S$
  and note that $K$ is a constraint core for
  \RDS and a solution core for \textsc{$r$-Scattered Set}
  for all $\lambda \leq r \leq \mu$.

  Let~$\hat G$ be a $(\mu,1)$-projection kernel of~$(G,K)$ constructed according
  to Lemma~\ref{lemma:projkernel}, recall that~$\hat G$ is an induced subgraph of~$G$
  with~$|\hat G| = O(|K|)$. Let~$O := V(\hat G)\setminus K$ be the vertices outside
  the core set~$K$. By construction, note that any minimal $r$-dominator of~$K$ in $\hat G$
  has size $\ds_r(G)$ and that any maximal $r$-scattered set of~$\hat G$ contained in~$K$
  as size $\scat_r(G)$ for all  $\lambda \leq r \leq \mu$.

  Let~$\sigma$ be an integer divisible by all integers~$(2r+1)$ for  $\lambda \leq r \leq \mu$.
  We construct~$G'$ from~$\hat G$ by appending a path of length~$\sigma-1$ to every
  vertex~$v \in O$ and call the resulting path (including~$v$) $P_v$.
  The size of~$G'$ is bounded by~$O(|K|) = O(\ds_\lambda(G))$,
  it remains to show the second property.

  Fix~$r \in [\mu,\lambda]$ and define~$c_r := \frac{\sigma}{2r+1}|O|$.
  First assume that~$D$ is a minimal $r$-dominating set of~$G$. Then, as
  in the proof of Theorem~\ref{thm:domination-bikernel}, there exists a set~$\hat D$
  of the same size that $r$-dominates $K$ in~$\hat G$. We construct~$D'$ from
  $\hat D$ by including~$\sigma/(2r+1)$ vertices of each path~$P_v$; namely all vertices
  at position $i(2r+1)-r$, $1 \leq i \leq \sigma/(2r+1)$ in $P_v$ (where~$v$ has position 1).
  Since these vertices dominate all of~$P_v$, we conclude that~$D'$ dominates all
  of~$G'$ and has size~$|D'| = |\hat D| + \frac{\sigma}{2r+1}|O| = |D| + c_r$.

  In the other direction, let~$D'$ be a minimum $r$-dominating set for~$G'$.
  Collect the vertices of~$D'$ that lie on~$P_v\setminus \{v\}$ in the set~$D'_P$.
  By a simple exchange argument $D'_P$ intersects every path~$P_v$
  in the same indices as above, \ie the vertices at positions $i(2r+1)-r$, $1 \leq i \leq \sigma/(2r+1)$.
  It follows that~$|D'_P| = c_r$. Note that~$D'_P$ cannot $r$-dominate any vertex
  in~$K$, hence~$D:= D' \setminus D'_P$ must $r$-dominate all of~$K$ and
  by construction of~$G'$ this also holds true in the graph~$\hat G$. Since~$K$
  is a constraint core for~$G$, we conclude that~$D$ $r$-dominates all of~$G$
  and has size~$|D| = |D'| - |D'_P| = |D'| - c_r$.
  We conclude that indeed $\ds_r(G') = \ds_r(G) + c_r$.

  Now consider a maximal $r$-scattered set~$I$ of~$G$. Then, as
  in the proof of Theorem~\ref{thm:scattered-bikernel}, there exists a set~$\hat I \subseteq K$
  which is $r$-scattered in~$\hat G$. We construct $I'$ from~$\hat I$
  by including~$\sigma/(2r+1)$ vertices of each path~$P_v$; namely all vertices
  at position $i(2r+1)$, $1 \leq i \leq \sigma/(2r+1)$ in $P_v$. Since~$\hat I$
  is disjoint with~$O$, the resulting set is indeed $r$-scattered and has size
  $|\hat I| := |I| + \frac{\sigma}{2r+1}|O| = |I| + c_r$.

  In the other direction, let $I'$ be a maximal $r$-scattered set in~$G'$.
  By a simple exchange argument, we can assume that~$I'$ contains all
  endpoints of the paths~$P_v$, $v \in O$ and, by repeating this argument,
  we can assume that~$I'$ intersects every path~$P_v$ at precisely the
  positions $i(2r+1)$, $1 \leq i \leq \sigma/(2r+1)$. Collect this part of~$I'$
  in the set~$I'_P$, note that~$|I'_P| = c_r$. We further conclude  that
  $O \cap I' = \emptyset$, therefore the set~$I := I'\setminus I'_P$ is completely
  contained in~$K$ and $I$ is $r$-scattered in~$\hat G$. Since~$K$ is a solution
  core, it follows that~$K$ is also~$r$-scattered in~$G$ and we have that
  $|I| = |I'| - |I'_P| = |I'| - c_r$.
  We conclude that indeed $\scat_r(G') = \scat_r(G) + c_r$ and therefore
  $\is_{2r}(G') = \is_{2r}(G) + c_r$.
\end{proof}

\section{Conclusion}

We defined the notion of \emph{water lilies} and showed that in BE-classes
these structures can be used to compute linear-sized cores, bikernels,
and BE-kernels. These constructions are almost universal, to the point were
we can combine them into `multikernels'. It stands to reason that there might
be a general formulation for these types of kernels. As a technical step, we
also prove that \RCDS admits a constant-factor approximation in BE-classes.

We are certain that our techniques directly translate to nowhere dense
classes but leave this endeavour as future work. Given that the problems
treated here all have constraints whose boundaries form intervals,
we ask whether the following artificial problem admits a polynomial
kernel in BE-classes: find a set~$D$ of size at most~$k$ such that
$|N^r[v] \cap D| \not \in \{0, 2\}$.

\bibliographystyle{plain}
\bibliography{biblio}

\begin{thebibliography}{10}

\bibitem{DSKernelPlanar}
Jochen Alber, Michael~R. Fellows, and Rolf Niedermeier.
\newblock Polynomial-time data reduction for dominating set.
\newblock {\em J. {ACM}}, 51(3):363--384, 2004.

\bibitem{DSKernelHtopPoly}
Noga Alon and Shai Gutner.
\newblock Kernels for the dominating set problem on graphs with an excluded
  minor.
\newblock {\em Electron. Colloquium Comput. Complex.}, 15(066), 2008.

\bibitem{DSKernelGenus}
Hans~L. Bodlaender, Fedor~V. Fomin, Daniel Lokshtanov, Eelko Penninkx, Saket
  Saurabh, and Dimitrios~M. Thilikos.
\newblock ({M}eta) {K}ernelization.
\newblock {\em J. {ACM}}, 63(5):44:1--44:69, 2016.

\bibitem{cDomInd}
Mustapha Chellali, Odile Favaron, Adriana Hansberg, and Lutz Volkmann.
\newblock k-domination and k-independence in graphs: {A} survey.
\newblock {\em Graphs Comb.}, 28(1):1--55, 2012.

\bibitem{DowneyFellows}
Rodney~G. Downey and Michael~R. Fellows.
\newblock {\em Parameterized Complexity}.
\newblock Monographs in Computer Science. Springer, 1999.

\bibitem{DSKernel}
P{\aa}l~Gr{\o}n{\aa}s Drange, Markus~Sortland Dregi, Fedor~V. Fomin, Stephan
  Kreutzer, Daniel Lokshtanov, Marcin Pilipczuk, \Michal Pilipczuk, Felix
  Reidl, Fernando {Sanchez Villaamil}, Saket Saurabh, Sebastian Siebertz, and
  Somnath Sikdar.
\newblock Kernelization and sparseness: the case of dominating set.
\newblock In {\em 33rd Symposium on Theoretical Aspects of Computer Science,
  {STACS} 2016}, volume~47 of {\em LIPIcs}, pages 31:1--31:14. Schloss Dagstuhl
  - Leibniz-Zentrum f{\"{u}}r Informatik, 2016.

\bibitem{DvorakDomset}
\Zdenek \Dvorak.
\newblock Constant-factor approximation of the domination number in sparse
  graphs.
\newblock {\em Eur. J. Comb.}, 34(5):833--840, 2013.

\bibitem{DvorakDomset2}
\Zdenek \Dvorak.
\newblock On distance-dominating and-independent sets in sparse graphs.
\newblock {\em Journal of Graph Theory}, 91(2):162--173, 2019.

\bibitem{DSKernelApexMinor}
Fedor~V. Fomin, Daniel Lokshtanov, Saket Saurabh, and Dimitrios~M. Thilikos.
\newblock Bidimensionality and kernels.
\newblock In {\em Proceedings of the Twenty-First Annual {ACM-SIAM} Symposium
  on Discrete Algorithms, {SODA} 2010}, pages 503--510. {SIAM}, 2010.

\bibitem{DSKernelHMinor}
Fedor~V. Fomin, Daniel Lokshtanov, Saket Saurabh, and Dimitrios~M. Thilikos.
\newblock Linear kernels for (connected) dominating set on \emph{H}-minor-free
  graphs.
\newblock In {\em Proceedings of the Twenty-Third Annual {ACM-SIAM} Symposium
  on Discrete Algorithms, {SODA} 2012}, pages 82--93. {SIAM}, 2012.

\bibitem{DSKernelHTopMinor}
Fedor~V. Fomin, Daniel Lokshtanov, Saket Saurabh, and Dimitrios~M. Thilikos.
\newblock Kernels for (connected) dominating set on graphs with excluded
  topological minors.
\newblock {\em {ACM} Trans. Algorithms}, 14(1):6:1--6:31, 2018.

\bibitem{DSKernelHtopPoly2}
Shai Gutner.
\newblock Polynomial kernels and faster algorithms for the dominating set
  problem on graphs with an excluded minor.
\newblock In {\em Parameterized and Exact Computation, 4th International
  Workshop, {IWPEC} 2009, Revised Selected Papers}, volume 5917 of {\em Lecture
  Notes in Computer Science}, pages 246--257. Springer, 2009.

\bibitem{BEGenColouring}
Stephan Kreutzer, \Michal Pilipczuk, Roman Rabinovich, and Sebastian Siebertz.
\newblock The generalised colouring numbers on classes of bounded expansion.
\newblock In {\em 41st International Symposium on Mathematical Foundations of
  Computer Science, {MFCS} 2016}, volume~58 of {\em LIPIcs}, pages 85:1--85:13.
  Schloss Dagstuhl - Leibniz-Zentrum f{\"{u}}r Informatik, 2016.

\bibitem{DSKernelND}
Stephan Kreutzer, Roman Rabinovich, and Sebastian Siebertz.
\newblock Polynomial kernels and wideness properties of nowhere dense graph
  classes.
\newblock {\em {ACM} Trans. Algorithms}, 15(2):24:1--24:19, 2019.

\bibitem{BndExp}
Jaroslav \Nesetril and Patrice {Ossona de Mendez}.
\newblock Grad and classes with bounded expansion i. decompositions.
\newblock {\em Eur. J. Comb.}, 29(3):760--776, 2008.

\bibitem{Sparsity}
Jaroslav \Nesetril and Patrice {Ossona de Mendez}.
\newblock {\em {Sparsity: Graphs, Structures, and Algorithms}}, volume~28 of
  {\em Algorithms and Combinatorics}.
\newblock Springer, 2012.

\bibitem{ISKernel}
\Michal Pilipczuk and Sebastian Siebertz.
\newblock Kernelization and approximation of distance-$r$ independent sets on
  nowhere dense graphs.
\newblock {\em arXiv preprint arXiv:1809.05675}, 2018.

\bibitem{NBComplexity}
Felix Reidl, Fernando~S{\'{a}}nchez Villaamil, and Konstantinos~S.
  Stavropoulos.
\newblock Characterising bounded expansion by neighbourhood complexity.
\newblock {\em Eur. J. Comb.}, 75:152--168, 2019.

\bibitem{WColBndExp}
Xuding Zhu.
\newblock Colouring graphs with bounded generalized colouring number.
\newblock {\em Discret. Math.}, 309(18):5562--5568, 2009.

\end{thebibliography}

\end{document}